\RequirePackage{xifthen}
\newboolean{ConferenceVersion}
\setboolean{ConferenceVersion}{false}
\ifthenelse{\boolean{ConferenceVersion}}
{
	\newcommand{\inConference}[1]{#1}
	\newcommand{\inFull}[1]{}
}
{
	\newcommand{\inConference}[1]{}
	\newcommand{\inFull}[1]{#1}
}

\documentclass[a4paper,UKenglish,autoref,thm-restate]{lipics-v2021}

\usepackage[ruled,vlined,linesnumbered,commentsnumbered]{algorithm2e}
\usepackage{nicefrac}
\usepackage{forloop}
\usepackage{listings}
\nolinenumbers
\hideLIPIcs

\title{Practical Budgeted Submodular Maximization}

\author{Moran Feldman}   {University of Haifa, Israel}{moranfe@cs.haifa.ac.il}{0000-0002-1535-2979}{}
\author{Zeev Nutov}          {The Open University of Israel, Ra'anana, Israel}{nutov@openu.ac.il}{}{}
\author{Elad Shoham}        {The Open University of Israel, Ra'anana, Israel}{shoham.elad@gmail.com}{}{}

\authorrunning{M. Feldman, Z. Nutov, and E. Shoham}

\Copyright{Moran Feldman, Zeev Nutov, and Elad Shoham}

\ccsdesc{Theory of computation~Approximation algorithms analysis}
\ccsdesc{Mathematics of computing~Combinatorial optimization}

\begin{document}

\maketitle

\newcommand {\ignore} [1] {}

\newcommand{\BSM}{{\textsc{BSM}}}
\newcommand{\BSMFull}{{\textsc{Budgeted Submodular Maximization}}}
\def\Gr  {\sc Greedy}
\def\PG  {\sc Plain Greedy}
\def\GP  {\sc Greedy$^+$}
\def\OG {\sc 1-Guess}
\def\TG  {\sc 2-Guess}
\def\ThG  {\sc 3-Guess}

\def\al     {\alpha}
\def\be     {\beta}
\def\ga     {\gamma}
\def\eps  {\varepsilon}
\def\th     {\theta}

\newcommand{\defcal}[1]{\expandafter\newcommand\csname c#1\endcsname{{\mathcal{#1}}}}
\newcommand{\defbb}[1]{\expandafter\newcommand\csname b#1\endcsname{{\mathbb{#1}}}}
\newcounter{calBbCounter}
\forLoop{1}{26}{calBbCounter}{
    \edef\letter{\Alph{calBbCounter}}
    \expandafter\defcal\letter
		\expandafter\defbb\letter
}

\def\sem  {\setminus}
\def\subs {\subseteq}
\def\f       {\frac}
\newcommand{\nnR}{{\bR_{\geq 0}}}
\newcommand{\pR}{{\bR_{> 0}}}

\keywords{submodular function, knapsack constraint, practical approximation algorithms}

\begin{abstract}
We consider the problem of maximizing a non-negative monotone submodular function subject to a knapsack constraint, 
which is also known as the {\BSMFull} (\BSM) problem. 
Sviridenko~\cite{Svir} showed that by guessing $3$ appropriate elements of an optimal solution, and then executing a greedy algorithm, 
one can obtain the optimal approximation ratio of $\al=1-\nicefrac{1}{e} \approx 0.632$ for {\BSM}. 
However, the need to guess (by enumeration) $3$ elements makes the algorithm of~\cite{Svir} impractical 
as it leads to a time complexity of roughly $O(n^5)$ (this time complexity can be slightly improved using the thresholding technique of Badanidiyuru \& Vondr\'{a}k~\cite{BV}, but only to roughly $O(n^4)$). Our main results in this paper show that fewer guesses suffice. 
Specifically, by making only $2$ guesses, 
we get the same optimal approximation ratio of $\alpha$ with an improved time complexity of roughly $O(n^3)$. 
Furthermore, by making only a single guess, we get an almost as good approximation ratio of $0.6174 > 0.9767\alpha$ in roughly $O(n^2)$ time.

Prior to our work, the only approximation algorithms that were known to obtain an approximation ratio close to $\alpha$ for {\BSM} 
were the algorithm of Sviridenko~\cite{Svir} and an algorithm of Ene \& Nguyen~\cite{EN} that achieves $(\al-\eps)$-approximation. 
However, the algorithm of~\cite{EN} requires ${(1/\eps)}^{O(1/\eps^4)}n\log^2 n$ time, and hence, is of theoretical interest only since 
${(1/\eps)}^{O(1/\eps^4)}$ is huge even for moderate values of $\eps$. In contrast, all the algorithms we analyze 
are simple and parallelizable, which makes them good candidates for practical use.

Recently, Tang et al.~\cite{TTLH} studied a simple greedy algorithm that already has a long research history, 
and proved that it admits an approximation ratio of at least $0.405$ (without any guesses). 
The last part of this paper improves over the result of~\cite{TTLH}, 
and shows that the approximation ratio of this algorithm is within the range $[0.427, 0.462]$.
%
\end{abstract}
\section{Introduction} \label{s:intro}

The last two decades have seen an impressive advancement in the theoretical understanding of submodular functions maximization. Furthermore, much of this advancement translated into improved results for applications domains such as non-parametric learning~\cite{mirzasoleiman2016distributed}, regression under human assistance~\cite{de2020regression}, interpreting neural networks~\cite{elenberg2017streaming}, adversarial attacks~\cite{lei2019discrete}, data summarization~\cite{kazemi2018scalable,mitrovic2018data}, fMRI parcellation~\cite{salehi2017submodular} and DNA sequencing~\cite{libbrecht2018choosing}, to name a few.

Naturally, early theoretical works on submodular maximization considered simple algorithms for simple kinds of constraints such as a single cardinality or knapsack constraint. Later works then extended the research field by considering more involved algorithms such as Continuous Greedy~\cite{calinescu2011maximizing} and non-oblivious local search~\cite{filmus2014monotone}; or more involved constraints such as intersection of $k$ matroids~\cite{lee2010submodular}, a constant number of knapsack constraints~\cite{kulik2013approximations} and $k$-exchange systems~\cite{feldman2011improved,ward2012approximation}. However, to this day, the early work on simple algorithms for simple constraints is the one that had the largest contribution to application domains due to two reasons. First, practitioners prefer to employ algorithms that are easy to implement and are practically fast (as opposed to just having a polynomial time complexity -- as is often the case with advanced theoretical algorithms). Second, simple (and in particular greedy) algorithms are often easily adaptable to less traditional models of computation such distributed computing and the data stream model.\footnote{We refer the reader to~\cite{YZA} for examples of such adaptations of an algorithm called {\GP} which plays a central part in this paper and is introduced below.}

Motivated by the above observation, we study in this work simple greedy algorithms for the problem of maximizing a 
non-negative monotone (increasing)  submodular function subject to a single knapsack constraint. 
In the following we refer to the last problem as {\BSMFull}, or {\BSM} for short; 
we refer the reader to Section~\ref{s:notation} for a more formal definition of this problem. 
An important special case of {\BSM} is the case in which all elements have the same cost, which makes the knapsack constraint simply a cardinality constraint. 
Already in $1978$, Nemhauser, Wolsey \& Fisher~\cite{NWF} showed that a simple greedy algorithm achieves 
an approximation ratio of $1-\nicefrac{1}{e}$ for this variant, 
which is tight (up to lower order terms) due to an inapproximability result 
published in the same year by Nemhauser \& Wolsey~\cite{nemhauser1978best}.


The natural extension of the greedy algorithm of Nemhauser et al.~\cite{NWF} to the general case of {\BSM} 
is an algorithm that we refer to in this paper as {\PG}. 
This algorithm starts with an empty solution, and then grows this solution in iterations. In each iteration {\PG} adds to its solution an element that 
(i) does not violate feasibility, and (ii)~among all elements obeying (i) has a maximum {\em density} -- 
that is, the quotient between the element's contribution to the objective function and its cost.
Section~\ref{s:notation} includes a pseudocode of {\PG} and the other algorithms mentioned in the current section. 

While {\PG} does not have any constant approximation ratio for general {\BSM}, several variants of it have been shown 
to possess such an approximation guarantee, and we survey them below. 
In this survey it will be useful to express approximation ratios as fractions of $\alpha = 1 - \nicefrac{1}{e}$ (one can recall that, by the hardness result of~\cite{nemhauser1978best}, an approximation ratio of $\alpha + o(1)$ is the best we can hope for). The reason that {\PG} fails to guarantee a constant approximation ratio is that most of the value of the optimal solution might belong to a single element $u$, and {\PG} might start by picking low value (but high density) elements that prevent $u$ from being added to the solution later (note that this scenario cannot occur in the case of uniform costs).
This observation naturally suggests an algorithm that we simply term {\Gr}, which outputs the best among the solution of {\PG} and the best singleton feasible solution. It is relatively easy to argue that {\Gr} guarantees $0.5\alpha$-approximation 
(see, for example, Khuller et al.~\cite{KMN} who proved this for a special case of {\BSM} known as {\sc Budgeted Coverage}), and already in $1982$, Wolsey~\cite{W-B} showed that {\Gr} in fact guarantees $0.35 \approx 0.554 \alpha$-approximation.\footnote{Technically, Wolsey~\cite{W-B} proved this approximation ratio for another variant of {\PG} that outputs either the solution of {\PG} or a particular singleton calculated based on the execution of {\PG}. However, since {\Gr} considers all feasible singletons as possible solutions, it is at least as good as the algorithm of~\cite{W-B}.}

Khuller, Moss \& Naor~\cite{KMN} also showed that for {\sc Budgeted Coverage} (the above mentioned special case of {\BSM}), guessing $3$ elements of the optimal solution that contribute the most to this solution (by iterating over all subsets of size $3$) and then executing {\PG} on the residual instance achieves the optimal approximation ratio of $\alpha \approx 0.632$. In this paper we denote the combination of {\PG} and guessing $k$ elements of the optimal solution by {\sc $k$-Guess \PG}. Sviridenko~\cite{Svir} adapted the analysis of Khuller et al.~\cite{KMN} to the general case of {\BSM}, and showed that {\ThG~\PG} obtains the optimal approximation ratio of $\alpha$ also for this case.
Unfortunately, however, the {\ThG~\PG} algorithm is very slow. Sviridenko~\cite{Svir} provided a na\"{i}ve implementation of this algorithm whose time complexity is as high as $O(n^5)$.\footnote{In the time complexity analysis, it is standard practice to assume that $f$ can be evaluated on any given set in constant time.} A somewhat more involved implementation can be obtained using the thresholding technique of~\cite{BV}; this implementation has the the same approximation ratio as the original algorithm (up to a factor $1 - \eps)$, but enjoys a reduced time complexity of $\tilde{O}(n^4/\eps)$, where the $\tilde{O}$ notation disregards poly-logarithmic terms.\footnote{We refer the reader to~\cite{YZA} for an example of the application of the thresholding technique to a greedy algorithm in the context of {\BSM}. This technique can be applied to all the greedy algorithms considered in this paper in a similar way.}

The high time complexity of the {\ThG~\PG} algorithm severely restricts its practical use, raising the need for faster algorithms for {\BSM} that still guarantee an approximation ratio close to the optimal ratio of $\alpha$. Prior to this work, the only work (we are aware of) that achieved this goal was the work of Ene \& Nguyen~\cite{EN} that guarantees an approximation ratio of $\alpha - \eps$ in time ${(1/\eps)}^{O(1/\eps^4)}n\log^2 n$.\footnote{Badanidiyuru \& Vondr\'{a}k~\cite{BV} also claimed, as one of their many results, an algorithm that is faster then {\ThG~\PG}, but an error was found in their analysis (see~\cite{EN} for details).} The algorithm of~\cite{EN} was an amazing theoretical breakthrough, but its practical significance is very low both because it is quite involved, and more importantly, because its time complexity is huge even for moderate values of $\eps$ (note that even for a relatively large value of $\eps$ such as $\eps = 0.1$ we have $(1/\eps)^{1/\eps^4} = 10^{10000}$).

In this work we achieve a speed up over {\ThG~\PG} using a different approach. 
Instead of taking advantage of complex tools such as the multilinear extension used by~\cite{EN}, 
we strive to reduce the number of elements of the optimal solution that need to be guessed, 
while keeping the algorithm used after the guessing a simple greedy algorithm. 
Our first result shows that two guesses suffice to achieve the {\em tight} approximation ratio of $\alpha$. 
Namely, we prove the guarantees that appear in the next theorem for the {\TG~\PG} algorithm,
that guesses the two most valuable elements of an optimal solution by enumerating over all the options,
and then executes {\PG} on the residual instance. 
\begin{restatable}{theorem}{TwoGuessTheorem} \label{t:1}
{\TG~\PG} admits an approximation ratio $\alpha \approx 0.632$ in $O(n^4)$ time. 
Furthermore, using the thresholding technique of~\cite{BV}, 
the time complexity of the algorithm can be reduced to $\tilde{O}(n^3/\eps)$ 
at the cost of worsening the approximation ratio by a factor of $1 - \eps$.
\end{restatable}

Arguably, {\TG~\PG} is currently the most practical algorithm guaranteeing the tight $\alpha$-approximation for {\BSM}, 
especially if one takes into account the observation that the enumeration step, used to implement the guessing, 
makes the algorithm highly amenable to parallelization. 
Nevertheless, if one is willing to make a small sacrifice in the approximation ratio, 
it turns out that a single guess is almost as good as two.

To formalize the last claim, we need to introduce an algorithm called {\GP} due to Yaroslavtsev, Zhou \& Avdiukhin~\cite{YZA}. This algorithm runs {\PG} as usual, but then outputs the best out of all the feasible solutions that can be obtained by augmenting any of the intermediate solutions of {\PG} with up to a single additional element. The description of {\GP} might give the impression that it is much slower than {\PG}, but in fact the two algorithms have roughly the same run time since {\PG} also needs to calculate the value of the combination of its current solution with every single element to find the next element that should be added to the solution. The following theorem considers an algorithm called {\OG~\GP} that guesses an element of the optimal solution contributing the most to this solution, and then executes {\GP} on the residual instance. Note that, as promised, the theorem shows that {\OG~\GP} guarantees an approximation ratio which diviates from the optimal $\alpha$ ratio by less than $3\%$.
\begin{restatable}{theorem}{OneGuessTheorem} \label{t:2}
{\OG~\GP} admits an approximation ratio of $\f{3-\ln 4}{4-\ln 4}>0.6174>0.9767 \al$ and runs in $O(n^3)$ time. Furthermore, using the thresholding technique of~\cite{BV}, the time complexity of of the algorithm can be reduced to $\tilde{O}(n^2/\eps)$ at the cost of worsening the approximation ratio by a factor of $1 - \eps$.
\end{restatable}


Since the {\Gr} algorithm is highly natural and has a long research history, a recent paper by Tang et al.~\cite{TTLH} suggested 
studying what is the best approximation ratio that can be proved for this algorithm. 
In retrospect, some answer for this question was given by Cohen \& Katzir~\cite{CK}, 
who showed that {\Gr} guarantees a ratio of at least 
$\f{\al}{1+\al} > 0.61 \al \approx 0.385$.\footnote{To be precise, Cohen \& Katzir~\cite{CK} considered a special case of {\BSM}, 
and proved their result only for this special case. However, the proof extends to the general case.}
Tang et al.~\cite{TTLH} themselves further improved the analysis of {\Gr} to guarantee $0.405 \approx 0.64\alpha$-approximation. 
Our final result, given by the next theorem, identifies the exact approximation ratio of {\Gr} up to an error of $0.035$. 
In particular, it shows that {\Gr} is strictly worse than more modern algorithms such as {\GP} (Yaroslavtsev et al.~\cite{YZA} showed that {\GP} achieves $\nicefrac{1}{2}$-approximation on its own) and the algorithms suggested by the current work.
\begin{restatable}{theorem}{NoGuessTheorem} \label{t:3}
The approximation ratio of {\Gr} is in the range $[0.427, 0.462]$, namely,
it is at least $0.427 > 0.675 \al$ and is no better than $0.462 \approx 0.73 \alpha$.
\end{restatable}

\subparagraph*{Paper Structure.} Section~\ref{s:notation} formally defines {\BSM} and the notation we use. 
This section also includes the pseudocode for all the algorithms we analyze in this paper, and a basic analysis used by later sections.
The proofs of our results are split between the next three sections. Section~\ref{s:2_guess} analyzes {\TG~\PG} and proves \autoref{t:1}. The algorithm {\OG~\GP} is analyzed in Section~\ref{s:1_guess}, which proves Theorem~\ref{t:2}. Finally, the analysis of {\Gr} and the proof of Theorem~\ref{t:3} appear in Section~\ref{s:no_guess}.
\section{Notation and Pseudocode of Algorithms} \label{s:notation}

In this section we formally define the notation we use and the problem {\BSM}. Using these definitions, we then give the pseudocode for the algorithms that we analyze in this work, and prove a basic result about these algorithms that is used by the next sections.

\paragraph*{Definitions}

Given a ground set $V$, a set function $f \colon 2^V \to \bR$ assigns a number to every subset of $V$.
The function $f$ is \textbf{monotone} (or non-decreasing) if $f(S) \leq f(T)$ for every two sets $S \subs T \subseteq V$.
Furthermore, it is \textbf{submodular} if for every set $S \subseteq V$ and two elements $u,v \notin S$ we have
\[
	f(S \cup \{u\}) - f(S) \geq f(S \cup \{u,v\}) - f(S \cup \{v\}) \enspace . 
\]
Intuitively, a set function is submodular if the marginal contribution of an element to a set $S$ can only decrease when other elements are added to the set. For simplicity of notation, given an element $v \in V$ and a set $S \subseteq V$, we often use below $S+v$, $S-v$, $f(v)$ and $f(v \mid S)$ as shorthands for $S \cup \{v\}$, $S \sem \{v\}$, $f(\{v\})$ and $f(S + v) - f(S)$, respectively. The expression $f(v \mid S)$ defined in the previous sentence is called the \textbf{marginal contribution} of $v$ with respect to the set $S$. Occasionally, it is also useful to consider the marginal contribution of a set $T \subseteq V$ with respect to another set $S \subseteq V$, which is defined as $f(T \mid S) \triangleq f(T \cup S) - f(S)$.

In the {\BSMFull} problem (\BSM), we are given a non-negative monotone submodular function $f\colon 2^V \to \nnR$, a positive cost function $c\colon V \to \pR$ and a positive budget $B$. The objective of the problem is to find a set maximizing $f$ among the sets $S \subseteq V$  whose cost is at most $B$ (i.e., $c(S) \triangleq \sum_{v \in S} c(v) \leq B$). For simplicity, we assume that $c(v) \leq B$ for every element $v \in V$. Clearly, any element violating this assumption cannot be a part of any feasible solution, and thus, can be discarded. Additionally, we denote by $OPT$ an arbitrary optimal solution for this problem, and occasionally assume that $OPT$ contains at least two elements. One can verify that all the algorithms we consider return an optimal solution when this assumption is violated. Finally, to avoid visual clutter, unless otherwise is explicitly mentioned, we assume both $f(OPT) = 1$ and $B = 1$ (these assumptions are without loss of generality since one can scale the costs and values to obtain these equalities, and the behavior of the algorithms we consider is independent of such scaling).

\paragraph*{Pseudocode of Algorithms}

One can recall that {\PG} starts with the empty solution, and then adds in every iteration the element with the maximum density (with respect to the current solution) among the elements whose addition to the solution does not violate feasibility. Formally, the density of an element $v \in V$ with respect to a set $S \subs V - v$ is defined as $f(v \mid S)/c(v)$. The pseudocode of {\PG}, which uses this definition, appears as Algorithm~\ref{alg:greedy}. Algorithm~\ref{alg:greedy} also includes the pseudocodes of {\Gr} and {\GP}, which differ from {\PG} only in their last lines. The last line of {\Gr} returns the better solution among the output of {\PG} and the best singleton; while the last line of {\GP} returns the better solution among the output of {\PG} and the best feasible solution that can be obtained by combining any solution that {\PG} had at some iteration with a single additional element. We would like to stress that all three algorithms {\PG}, {\Gr} and {\GP} have the same asymptotic time complexity despite the different time complexities required by their respective last lines.

\begin{algorithm}
\caption{\PG~\slash~\Gr~\slash~\GP} \label{alg:greedy}
Let $i \gets 0$ and $S_i \gets \varnothing$.\\
\While{there exists an element $v \in V \setminus S_i$ such that $c(S_i + v) \leq B$\label{line:loop}}
{
	Let $s_{i + 1}$ denote the element $v$ maximizing $f(v \mid S_i)/c(v)$ among the elements satisfying the condition given in Line~\ref{line:loop}.\\
	Let $S_{i + 1} \gets S_i + s_{i + 1}$.\\
	Increase $i$ by $1$.
}
In {\PG}: \Return{$S_i$}.\\
In {\Gr}: \Return{the set maximizing $f$ in $\{S_i\} \cup \{\{v\} \mid v \in V\}$}.\\
In {\GP}: \Return{the set maximizing $f$ in $\{S_i\} \cup \{S_{i'} + v \mid 0 \leq i' < i, v \in V \text{ and } c(S_{i'} + v) \leq B\}$}.
\end{algorithm}

We often consider algorithms that guess $k$ elements of $OPT$, for some positive integer $k$, and then execute a given algorithm $ALG$ such as {\Gr} or {\GP} on the residual instance. A general template of an algorithm of this kind is given as Algorithm~\ref{alg:guess}. To handle correctly also cases in which $OPT$ contains less than $k$ elements, the algorithm also considers all feasible solutions consisting of less than $k$ elements as possible outputs. Informally, the \textbf{residual instance}, that we refer to above, is the instance of {\BSM} obtained by assuming that the guessed elements are implicitly added to the solution. One can observe that this informal definition is consistent with the construction of the instance on which Algorithm~\ref{alg:guess} executes $ALG$ given guess $Y$.

\begin{algorithm}
\caption{\textsc{$k$-Guess}~$ALG$} \label{alg:guess}
\DontPrintSemicolon
\For{every set $Y \subseteq \cN$ of size $k$ that obeys $c(Y) \leq B$}{
	Execute $ALG$ on the residual instance defined by the ground set $V \setminus Y$, the objective function $h(S) = f(S \mid Y)$ and the budget $B - c(Y)$.\\
	Add $Y$ to the output set of $ALG$ to get a feasible solution for the original problem. \label{line:add_Y}
}
Let $\cS_1$ be the set of solutions obtained in Line~\ref{line:add_Y} in any of the iterations of the loop.\\
Let $\cS_2$ be the set of all feasible solutions containing less than $k$ elements.\\
\Return{the best solution in $\cS_1 \cup \cS_2$}.
\end{algorithm}

It is also worth mentioning that running any of the algorithms {\PG}, {\Gr} or {\GP} on the residual instance corresponding to a guess $Y$ is equivalent to executing the same algorithm on the original instance, but starting from $Y$ as $S_0$ instead of the empty set.

\paragraph*{Basic Analysis}

We complete this section with a basic analysis that applies to all the greedy algorithms we consider. To present this analysis, let us denote by $\ell$ the number of iterations performed by \autoref{alg:greedy}, and let us define for every $i \in \{0, 1, \dotsc, \ell\}$, $g(c(S_i)) = f(S_i)$. In other words, given that \autoref{alg:greedy} spent an $x$ fraction of its budget after some number $i$ of iterations, $g(x)$ gives the fraction of the value of the optimal solution that \autoref{alg:greedy} had at this point.

Let us now denote by $\Delta_g(c(S_i))$ the \emph{rate} in which $g$ increases between $c(S_i)$ and $c(S_{i + 1})$, i.e., the ratio
\[
	\frac{g(c(S_{i + 1})) - g(c(S_{i}))}{c(S_{i + 1}) - c(S_{i})}
	=
	\frac{f(S_{i + 1}) - f(S_{i})}{c(S_{i + 1}) - c(S_{i})}
	=
	\frac{f(s_{i + 1} \mid S_i)}{c(s_{i + 1})}
\]
(note that this ratio is also the density of the element $s_{i + 1}$). The next simple lemma lower bounds $\Delta_g(c(S_i))$.\inConference{ Its proof is deferred to Appendix~\ref{app:missing_preliminaries}, but intuitively it follows by (i) showing that the average density of the elements of $S^* \setminus R$ with respect to $S_i \cup R$ is large, and then (ii) using the submodularity of $f$ to argue that one of these elements has a high density also with respect to $S_i$, and therefore, $s_{i + 1}$ also has such a high density.}
\begin{restatable}{lemma}{lr}  \label{l:r}
For every two sets $R \subs S^*$ and an integer $i=0,1,\ldots,\ell - 1$, if we have $c(S^*) \leq 1$ and $1 - c(S_i) \geq \max_{s \in S^* \sem (R \cup S_i)} c(s)$, then
\begin{align*}
 (1- c(R)) \cdot \Delta_g(c(S_i)) \geq{} & f(S^*) -f(S_i \cup R) \\ \geq{} & f(S^*)-g(c(S_i))-\sum_{s \in R \sem S_i} \mspace{-9mu} f(s) \geq f(S^*)-g(c(S_i))-\sum_{s \in R} f(s) \enspace .
\end{align*}
\end{restatable}
\newcommand{\lrProof}{
\begin{proof}
The second and last inequalities of the lemma follow from the following calculation, which holds by the submodularity of $f$, the definition of $g$ and the non-negativity of $f$.
\[
	f(S_i \cup R) \leq f(S_i)+\sum_{s \in R \sem S_i} f(s)=g(c(S_i))+\sum_{s \in R \sem S_i} f(s) \leq g(c(S_i))+\sum_{s \in R} f(s)\ .
\]

It remains to prove the first inequality of the lemma. To do that, it is useful to define $R_i=S_i \cup R$. Then,
\begin{align*}
	f(S^*)-f(S_i \cup R)
	={} &
	f(S^*)-f(R_i)
	\leq
	f(S^* \cup R_i)-f(R_i) \leq
	\sum_{s \in S^* \sem R_i} \mspace{-9mu} f(s \mid S_i)\\
	={} &
	\sum_{s \in S^* \sem R_i} \mspace{-9mu} c(s) \f{f(s \mid S_i)}{c(s)}
  \leq
	\sum_{s \in S^* \sem R_i} \mspace{-9mu} c(s) \f{f(s_{i + 1} \mid S_i)}{c(s_{i + 1})}\\
	={} &
	(c(S^*) - c(S^* \cap R_i)) \cdot \Delta_g(c(S_i))
	\leq
	(1 - c(R)) \cdot \Delta_g(c(S_i)) \enspace .
\end{align*}
The first inequality holds since $f$ is monotone.
The second inequality holds since $f$ is submodular. 
The third inequality follows from the way $s_{i + 1}$ is chosen by {\Gr} and the observation that the condition $1 - c(S_i) \geq  \max_{s \in S^* \sem (R \cup S_i)} c(s)$ implies that it is feasible to add any element $s \in S^* \sem R_i$ to $S_i$. 
Finally, the last inequality holds since $c(S^*) \leq 1$ and $R \subs S^* \cap R_i$. 
\end{proof}}\inFull{\lrProof}
\section{Analyzing {\TG~\PG} (Theorem~\ref{t:1})} \label{s:2_guess}

In this section we prove the approximation guarantee stated in \autoref{t:1}. We do not explicitly prove the time complexities stated in this theorem (or any of our other theorems) because they immediately follow from previous works such as~\cite{YZA}. For convenience, we repeat \autoref{t:1} below.
\TwoGuessTheorem*

We begin by the proof of the theorem with the following lemma regarding {\PG}.\inConference{ Since similar lemmata have been proved in many previous works, and due to space constraints, we defer the proof of this lemma to Appendix~\ref{app:two_guesses_missing}.}
\begin{restatable}{lemma}{lemOneElementLessGuarantee} \label{lem:one_element_less_guarantee}
Let $S^*$ be a feasible solution, and let $k$ be the smallest integer $1 \leq k \leq \ell - 1$ such that there is an element $w \in S^*$ such that $c(w) > B - c(S_k)$. If $k$ exists, then {\PG} produces a solution of value at least $\alpha \cdot [f(S^* - f(w)]$. Otherwise, {\PG} produces a solution of value at least $\alpha \cdot f(S^*)$.
\end{restatable}
\newcommand{\proofLemOneElementLessGuarantee}{\begin{proof}
We begin the proof by considering the case in which $k$ exists, and as usual assume $B = 1$. Furthermore, we also assume that $f(S_i) < f(S^*) - f(w)$ for every $0 \leq i \leq k - 1$, which is without loss of generality since the lemma follows immediately from the monotonicity of $f$ when this assumption is violated.

We now observe that the definition of $w$ implies that Lemma~\ref{l:r} applies for every $0 \leq i \leq k - 1$ and set $R \subseteq S^*$. Choosing $R = \{w\}$, we get
\[
	\Delta_g(c(S_i))
	\geq
	\frac{f(S^*) - g(c(S_i)) - f(w)}{1 - c(w)}
	=
	\frac{f(S^*) - f(w) - f(S_i)}{1 - c(w)}
	\quad
	\forall 0 \leq i \leq k - 1
	\enspace.
\]
Plugging the definition of $\Delta_g(c(S_i))$ into the last inequality now gives, for every $0 \leq i \leq k - 1$,
\[
	\frac{f(S_{i + 1}) - f(S_i)}{c(s_{i + 1})} \geq \frac{f(S^*) - f(w) - f(S_i)}{1 - c(w)}
	\enspace,
\]
and rearranging this inequality implies
\begin{align*}
	f(S^*) - f(w) - f(S_{i + 1})
	\leq{} &
	\left(1 - \frac{c(s_{i + 1})}{1 - c(w)}\right)[f(S^*) - f(w) - f(S_i)]\\
	\leq{} &
	e^{-c(s_{i + 1}) / (1 - c(w))} \cdot [f(S^*) - f(w) - f(S_i)]
	\enspace.
\end{align*}
Unraveling the last inequality for all $0 \leq i \leq k - 1$ gives,
\begin{align*}
	f(S^*) - f(w) - f(S_{k}&)
	\leq
	\prod_{i = 0}^{k - 1} e^{-c(s_{i + 1}) / (1 - c(w))} \cdot [f(S^*) - f(w) - f(S_0)]\\
	={} &
	e^{-c(S_{k}) / (1 - c(w))} \cdot [f(S^*) - f(w) - f(S_0)]
	\leq
	e^{-1} \cdot [f(S^*) - f(w)]
	\enspace,
\end{align*}
where the last inequality holds by the non-negativity and monotonicity of $f$ and the observation that the definition of $k$ implies $c(S_k) \geq 1 - c(w)$. The first case of the lemma now follows by rearranging the last inequality since the monotonicity of $f$ guarantees that $f(S_k)$ is a lower bound on the value of the output of {\Gr}.

The proof for the second case of the of the lemma is very similar to the proof of the first part. The only two changes that need to be done are: (i) $R$ should be chosen as the empty set, and (ii) $w$ should be chosen as a dummy element of cost $0$ that does not affect the objective function $f$ at all.
\end{proof}}
\inFull{\proofLemOneElementLessGuarantee}

Let us now denote by $\{u_1, u_2\}$ the maximum value subset of $OPT$ of size two (recall that we assume $|OPT| \geq 2$). We prove \autoref{t:1} by considering the iteration of {\TG~\PG} in which it guesses $Y = \{u_1, u_2\}$. In this iteration, {\TG~\PG} executes {\PG} on the residual instance defined by the ground set $V \setminus Y$, the budget $B - c(Y)$ and the objective function $h(S) = f(S \mid Y)$. Let us denote this residual instance by $\cI'$, and let us denote by $\ell'$ and $S'_i$ the values of $\ell$ and $S_i$, respectively, corresponding to the execution of {\PG} on the instance $\cI'$. Finally, let $k'$ be the smallest integer $1 \leq k' \leq \ell' - 1$ such that there exists an element $w' \in OPT \setminus \{u_1, u_2\}$ such that $c(w') > (1 - c(Y)) - c(S'_k)$.

The following lemma completes the proof of \autoref{t:1} (note that $Y \cup S'_{\ell'}$ is one of the solutions considered by {\TG~\PG} for its output).
\begin{lemma}
$f(Y \cup S'_{\ell'}) \geq \alpha \cdot f(OPT) = (1 - 1/e) \cdot f(OPT)$.
\end{lemma}
\begin{proof}
If $k'$ does not exist, then since $S^* = OPT \setminus Y$ is a feasible solution for the residual instance $\cI'$, Lemma~\ref{lem:one_element_less_guarantee} guarantees that the output $S'_{\ell'}$ of {\PG} is of value at least $\alpha \cdot h(S^*)$ according to the objective function $h$ of the residual instance. Therefore,
\begin{align*}
	f(Y \cup S'_{\ell'})
	={} &
	h(S'_{\ell'}) + f(Y)
	\geq
	\alpha \cdot h(OPT \setminus Y) + f(Y)\\
	={} &
	\alpha \cdot f(OPT) + (1 - \alpha) \cdot f(Y)
	\geq
	\alpha \cdot f(OPT)
	\enspace,
\end{align*}
where the last inequality follows from the non-negativity of $f$.

Consider now the case in which $k'$ exists, in this case Lemma~\ref{lem:one_element_less_guarantee} guarantees that the output $S'_{\ell'}$ of {\PG} is of value at least $\alpha \cdot [h(S^*) - h(w')]$ according to the objective function $h$ of the residual instance. Therefore,
\begin{align*}
	f(Y \cup S'_{\ell'})
	={} &
	h(S'_{\ell'}) + f(Y)
	\geq
	\alpha \cdot [h(OPT \setminus Y) - h(w')] + f(Y)\\
	={} &
	\alpha \cdot f(OPT) - \alpha \cdot f(w' \mid Y) + (1 - \alpha) \cdot f(Y)
	\enspace.
\end{align*}
To see that the last inequality completes the proof of the second case, we note that the definition of $Y$ and the submodularity of $f$ imply
\begin{align*}
	(1 - \alpha) \cdot f(Y)
	\geq{} &
	(1 - \alpha) \cdot [f(\{u_1, w\}) + f(\{u_2, w\}) - f(Y)]\\
	={}&
	(1 - \alpha) \cdot [f(w \mid \{u_1\}) + f(w \mid \{u_2\}) + f(u_1) + f(u_2) - f(Y)]\\
	\geq{} &
	2(1 - \alpha) \cdot f(w \mid Y)
	\geq
	\alpha \cdot f(w \mid Y)
	\enspace.
	\qedhere
\end{align*}
\end{proof}

\section{Analyzing \texorpdfstring{\OG~\GP}{\OG~\Gr+} (Theorem~\ref{t:2})} \label{s:1_guess}

%

In this section we prove the approximation guarantee stated in \autoref{t:2}. For convenience, we repeat the theorem itself below. Recall that {\GP} is a variant of {\PG} that considers as a possible output every solution that can be obtained by combining any intermediate solution $S_i$ of {\PG} with one other element.
\OneGuessTheorem*


\inFull{As a warm-up, we reprove below one of the main results of~\cite{YZA}. Let $r$ be an element of $OPT$ of maximum cost.}
\inConference{In the proof of \autoref{t:2} we use the following result of~\cite{YZA}. For completeness, we include a proof of this result in Appendix~\ref{app:G_plus_missing}.}

\begin{restatable}[Yaroslavtsev, Zhou \& Avdiukhin \cite{YZA}]{theorem}{tHalf} \label{t:1/2}
{\GP} guarantees an approximation ratio of $1/2$.
\end{restatable}
\newcommand{\proofTHalf}{\begin{proof}
Let $k$ be the smallest integer $1 \leq k \leq \ell$ for which $c(S_k) > 1 - c(r)$. If $k$ does not exist, then Lemma~\ref{lem:one_element_less_guarantee} guarantees that {\PG}, and therefore also {\GP}, achieves in fact an approximation ratio of $\alpha = 1 - 1/e > 1/2$. Thus, we assume below that $k$ exists.

Applying Lemma~\ref{l:r} with $R=\{r\}$ and $S^* = OPT$, we get for every $0 \leq i < k$,
\[
	1
	\leq
	(1 - c(r)) \cdot \Delta_g(c(S_{i})) + f(S_{i} + r)
	=
	(1 - c(r)) \cdot \frac{f(S_{i + 1}) - f(S_{i})}{c(s_{i + 1})} + f(S_{i} + r)
	\enspace.
\]
If $f(S_{i} + r) \geq \nicefrac{1}{2}$ for any integer $0 \leq i < k$, then we are done because $S_i + r$ is one of the solutions considered by {\GP}. Otherwise, rearranging the last inequality yields
\[
	\frac{c(s_{i + 1})}{2(1 - c(r))}
	\leq
	f(S_{i + 1}) - f(S_{i})
	\enspace.
\]
Adding up this inequality over all values of $i$ gives
\[
	f(S_k)
	\geq
	\sum_{i = 0}^{k - 1} \frac{c(s_{i + 1})}{2(1 - c(r))} + f(S_0)
	>
	\frac{1}{2}
	\enspace,
\]
where the second inequality follows from the non-negativity of $f$ and the fact that by the definition of $k$ we have $\sum_{i = 0}^{k - 1} c(s_{i + 1}) = c(S_k) > 1 - c(r)$. The theorem follows since $f(S_k)$ is a lower bound on the value of the output of {\GP} by the monotonicity of $f$.
%
\end{proof}}
\inFull{\proofTHalf}

\inConference{Let $r$ be an element of $OPT$ of maximum cost. }The guarantee of \autoref{t:1/2} is completely independent of the properties of the element $r$. However, when $r$ has a small value, we intuitively expect the guarantee of {\GP} to improve because in this regime $OPT - r$ is a solution of high value whose individual elements can still be added to the solution of {\PG} until this solution reaches a cost of at least $1 - c(r)$. Lemma~\ref{l:1-z} below formally states such an improved guarantee for {\PG} in the case of a small $f(r)$, but before getting to this lemma we need to present some preliminaries.

\begin{figure}
\centering 
\includegraphics{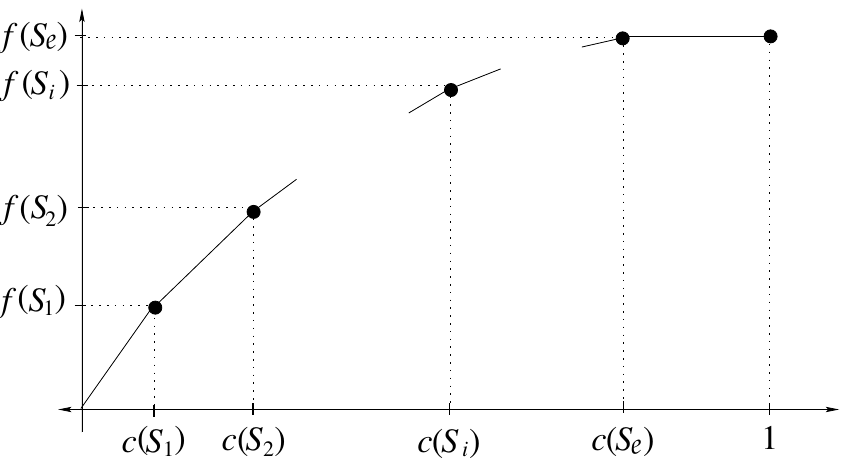}
\caption{Illustration of the piecewise linear extension of $g$.}
\label{f:G+}
\end{figure}

As defined above, the function $g$ is defined only for points $t \in [0, 1]$ that happen to have $t = c(S_i)$ for some $0 \leq i\leq \ell$. We extend it to all points of $[0, 1]$ by (i) defining $g(1) = f(S_\ell)$ and then (ii) connecting every two adjacent points in the graph of $g$ with a linear segment (see Figure~\ref{f:G+}). Formally, if we define $S_{\ell + 1} = S_\ell$ and let $i(t)$ be the largest integer such that $c(S_{i(t)}) \leq t$, then for every $t \in [0, 1]$ we have
\[
	g(t)
	=
	\frac{(t - c(S_{i(t)})) \cdot f(S_{i(t) + 1}) + (c(S_{i(t) + 1}) - t) \cdot f(S_{i(t)})}{c(s_{i + 1})} \enspace.
\]

Note that the function $g$ is of interest because the equality $f(S_\ell) = g(1)$, which holds by $g$'s definition, and our assumption that $f(OPT) = 1$ mean together that $g(1)$ is the approximation ratio {\GP}. In light of this observation, it is interesting to study some basic properties of the function $g$, which we do next.
\begin{observation}
The function $g$ is non-decreasing and continuous. Moreover, its derivative is defined at every point in the range $[0, 1]$ except for maybe a finite set of points, and at every point $t \in [0, 1 - c(r)]$ in which $g'(t)$ is defined we have
\begin{equation} \label{e:g}
	g'(t) \geq \max\left\{\f{1-g(t)-f(r)}{1-c(r)},\f{1-f(S_{i(t)}+r)}{1-c(r)} \right\} \ .
\end{equation}
\end{observation}
\begin{proof}
The non-decreasing and continuity properties of $g$ follow from the definition of $g$ and the monotonicity of $f$ because {\GP} only adds elements to its solution over time. Moreover, because $g$ is defined as the piecewise linear extension of its values in the points of the set $\{1\} \cup \{c(S_i) \mid 0 \leq i \leq \ell\}$, its derivative is well defined in every point within the range $[0, 1]$ except for maybe the points in the last set. Therefore, to prove the observation it only remains to show the lower bound on $g'(t)$ given by Inequality~\eqref{e:g}.

Consider some point $t \in [0, 1 - c(r)]$ in which the derivative $g'(t)$ is defined. If $t < c(S_\ell)$, then by choosing $S^* = OPT$ and $R = {r}$ in Lemma~\ref{l:r}, we get
\[
	g'(t)
	=
	\Delta_g(c(S_{i(t)}))
	\geq
	\f{1-f(S_{i(t)}+r)}{1-c(r)}
	\geq
	\f{1-g(c(S_{i(t)}))-f(r)}{1-c(r)}
	=
	\f{1-f(S_{i(t)})-f(r)}{1-c(r)} \enspace . 
\]
Consider now the case $t \geq c(S_\ell)$. Since $t \leq 1 - c(r)$, in this case every element of $OPT$ can be added to $S_\ell$ without violating feasibility (recall that $r$ is the costliest element in $OPT$). However, we also know that {\GP} terminated after $\ell$ iterations, and the only way in which these two observations can co-exist is when $OPT \subseteq S_{\ell}$. Therefore, the right hand side of Inequality~\eqref{e:g} is non-positive in this case, and the left hand side is non-negative since $g$ is non-decreasing.
\end{proof}

We need the following lemma, which defines an auxiliary function $z(y)$ playing a central part in the rest of this section. Since the proof of this lemma is mainly technical, we defer it to Appendix~\ref{app:G_plus_missing}.

\begin{restatable}{lemma}{lz} \label{l:z}
For every $y \in [0, 1/2]$, there is a unique value $z(y) \in [y, 1/2]$ satisfying the equation 
\[
	\f{y}{z(y)} - 1 = \ln\left(\frac{z(y)}{1 - y}\right) \enspace .
\]
Moreover, $z'(y) = \frac{z(y) (1 - y - z(y))}{(1 - y)(z(y) + y)}$,
$z(y)$ is a non-decreasing function of $y$, and $z(y) \geq z(0) = 1/e$.
\end{restatable}

We are now ready to give the promised approximation guarantee for {\GP} for the case in which the element $r$ has a (relatively) small value.

\begin{lemma} \label{l:1-z}
If $f(r) \leq f(OPT)/2$, then {\GP} admits ratio $1 - z(f(r))$.
\end{lemma}
\begin{proof}
For brevity, we use the shorthand $z = z(f(r))$ in the calculations below. 
Assume towards a contradiction that the lemma is false. 
In particular, this implies $f(S_{i(t)}+r) < 1 - z$ for every $t \in [0,1-c(r)]$.
Plugging this observation into (\ref{e:g}) yields, for every value $t \in [0,1-c(r)]$ for which $g'(t)$ is defined,
\[
	g'(t) \geq \max\left\{\frac{1 - f(r) - g(t)}{1 - c(r)}, \frac{z}{1 - c(r)}\right\} \enspace .
\]

Let us choose $t_s = -[1 - c(r)] \cdot \ln\left(\frac{z}{1 - f(r)}\right)$, 
where the subscript $s$ in $t_s$ stands for ``switch'' because we switch 
between the two lower bounds on $g'(t)$ at this value. By the definition of $z$,
\[
\ln\left(\frac{z}{1 - f(r)}\right) = \frac{f(r)}{z} - 1 \in [2f(r) - 1, 0] \subseteq [-1, 0]\enspace,
\]
where the membership holds since $z \in [f(r), 1/2]$ by Lemma~\ref{l:z}. 
Therefore, $t_s \in [0, 1 - c(r)]$, and the function $g(t)$ obeys the differential inequality 
$g'(t) \geq \frac{1 - f(r) - g(t)}{1 - c(r)}$ in every point within the range $[0, t_s]$.\footnote{Technically, this inequality holds for every value $t$ in this range, except for maybe a finite number of points in which $g'(t)$ is not defined. However, since $g$ is continuous, we can safely ignore this technical issue.}
The solution for the last inequality is $g(t) \geq [1 - f(r)] \cdot (1 - e^{-t/(1 - c(r))})$, which in particular implies 
\[
g(t_s) \geq [1 - f(r)] \cdot \left(1 - e^{-t_s / (1 - c(r))}\right) =  [1 - f(r)] \cdot \left(1-\frac{z}{1 - f(r)}\right)=1-f(r)-z
\enspace.
\]

The above inequality was obtained using one lower bound on $g'(t)$, and we now need to use the other lower bound. Specifically,
since $\f{zt_s}{1-c(r)}=-z \ln\left(\frac{z}{1-f(r)}\right)=z-f(r)$ by the definition of $z$, 
and $g'(t) \geq \f{z}{1-c(r)}$ for every value $t \in [t_s, 1 - c(r)]$ for which $g'(t)$ is defined,
\[
	g(1-c(r)) = g(t_s) + \int_{t = t_s}^{1 - c(r)} g'(t) dt \geq g(t_s) + (1-c(r)-t_s) \cdot \frac{z}{1-c(r)} \geq 1-z \enspace . 
\]
Note that the last inequality contradicts our assumption that {\GP} does not admit a ratio of $1 - z$ because the monotonicity of $f$ guarantees that the value of the output of {\GP} is at least $g(1) \geq g(1 - c(r)) \geq 1-z$.
\end{proof}

Up to this point we have considered {\GP}. Starting from this point we consider {\OG~\GP}, which is the algorithm to which \autoref{t:2} refers. Recall that we assume that $OPT$ contains at least two elements, and let us denote by $w$ the element of $OPT$ maximizing $f(w)$ and by $r'$ the element of $OPT - w$ with the maximum cost. The following two lemmata analyze the guarantee of {\OG~\GP} in two cases defined by the value of $f(w)$. \inFull{Each one of these lemmata is based on a different guarantee that was proved above for {\GP}.}\inConference{Due to space constraints, we defer the proofs of these lemmata to Appendix~\ref{app:G_plus_missing}. However, we note here that both proofs consider the iteration of {\OG~\GP} in which the set $Y$ of the guess contains exactly the element $w$, and then use one of the previous guarantees to lower bound the value of the output produced by {\GP} at this iteration with respect to the objective function of the residual instance. Specifically, the proof of \autoref{lem:large_w} uses Theorem~\ref{t:1/2} for that purpose, while the proof of \autoref{l:ph} employs \autoref{l:1-z} (and the observation that $f(r) \leq f(w)$ by the defintion of $w$).}


\begin{restatable}{lemma}{lemLargeW} \label{lem:large_w}
{\OG~\GP} admits a ratio of $\f{1+f(w)}{2}$, and in particular a ratio of at least $2/3$ whenever $f(w) \geq 1/3$.
\end{restatable}
\newcommand{\proofLemaLargeW}{\begin{proof}
Consider the iteration of {\OG~\GP} in which the set $Y$ of the guess contains exactly the element $w$. 
We note that $OPT-w$ is a feasible solution for the residual instance considered by this iteration, and the value of this
solution with respect to the objective function of this residual instance is $f(OPT - w \mid \{w\}) = f(OPT) - f(w) = 1 - f(w)$.
Hence, when {\GP} is applied by {\OG~\GP} to this residual instance, it produces a solution $S$ obeying
\[
	f(S \mid \{w\}) \geq \f{1 - f(w)}{2}
\]
because {\GP} admits a ratio of $1/2$ by Theorem~\ref{t:1/2}.
Consequently, the output $S + w$ constructed by {\OG~\GP} in the above mentioned iteration has a value of at least
\[
	f(S+w)=f(w)+ f(S \mid \{w\}) \geq f(w) + \f{1 - f(w)}{2} = \f{1 + f(w)}{2}
	\enspace.
	\qedhere
\]
\end{proof}}\inFull{\proofLemaLargeW}

\begin{restatable}{lemma}{lPh} \label{l:ph}
If $f(w) \leq 1/3$, then {\OG~\GP} achieves an approximation ratio of at least $p(f(r))$, where $p(\cdot)$ is the function
\[
p(x)\triangleq x+(1 -x) \cdot \left(1-z\left(\frac{x}{1 - x}\right)\right) \enspace .
\]
\end{restatable}
\newcommand{\proofLPh}{\begin{proof}
Like in the proof of the previous lemma, we consider in this proof the iteration of {\OG~\GP} in which the set $Y$ of the guess contains exactly the element $w$.
Furthermore, like in that proof, we note again that $OPT-w$ is a feasible solution for the residual instance considered by this iteration whose value,
with respect to the objective function of this residual instance, is $f(OPT - w \mid \{w\}) = f(OPT) - f(w) = 1 - f(w)$.

Since the definition of $w$ and the condition of the lemma imply together the inequality $\f{f(r')}{f(OPT - w)} \leq \f{f(w)}{f(OPT) - f(w))} \leq \f{1}{2}$,
Lemma~\ref{l:1-z} guarantees that in the considered iteration of {\OG~\GP} the output set $S$ of {\GP} obeys
\[
	\frac{f(S \mid \{w\})}{1 - f(w)}
	=
	\frac{f(S \mid \{w\})}{f(OPT - w \mid \{w\})}
	\geq
	1 - z\left(\f{f(r')}{f(OPT-w)}\right)  \geq 1-z\left(\f{f(w)}{1-f(w)}\right) \enspace .
\]
The last inequality holds since $z(\cdot)$ is a non-decreasing function by \autoref{l:z}.
Using this inequality, we get that the set $S + w$, which is one of the candidates considered by {\OG~\GP} for its output set, has a value of at least
\[
	f(S+w)=f(w) + f(S \mid \{w\})
	\geq
	f(w) + (1 - f(w)) \cdot \left(1-z\left(\f{f(w)}{1-f(w)}\right)\right)
	=
	p(f(w))
	\enspace.
	\qedhere
\]
\end{proof}}\inFull{\proofLPh}


Due to the use of the (quite complex) function $p(x)$, it is difficult to understand the guarantee of the last lemma. The following lemma shows that, within the relevant range, $p(x)$ is always at least $\f{3-\ln 4}{4-\ln 4}$. Since the proof of this lemma is technical, we defer it \inConference{also }to Appendix~\ref{app:G_plus_missing}. Additionally, we note that Theorem~\ref{t:2} is obtained immediately by combining Lemma~\ref{lem:P} with Lemmata~\ref{lem:large_w} and~\ref{l:ph}.

\begin{restatable}{lemma}{lemP} \label{lem:P}
$\min\{p(x):x \in [0,1/3]\}=\f{3-\ln 4}{4-\ln 4}$.
\end{restatable}

\section{Analyzing {\Gr} (Theorem~\ref{t:3})} \label{s:no_guess}

In this section we analyze the approximation ratio of {\Gr}, and prove Theorem~\ref{t:3}, which is restated below for convenience. The lower bound stated in the theorem is proved in Section~\ref{ssc:lower_bound_greedy}, and the upper bound stated is proved in Section~\ref{ssc:upper_bound_greedy}.
\NoGuessTheorem*


\subsection{Lower Bounding {\Gr}'s Approximation Ratio} \label{ssc:lower_bound_greedy}

In this section we prove the lower bound on the approximation ratio of {\Gr} stated in \autoref{t:3}. We use again the continuous version of the function $g$ introduced in Section~\ref{s:1_guess}. Furthermore, like in the last section, our proof is focused on showing that this function grows relatively quickly.

The next observation provides some lower bounds on the rate in which the discrete version of the function $g$ grows. Corollary~\ref{cor:greedy_continuous_derivative} later converts these bounds into guarantees for the continuous version of $g$\inConference{ (due to space constraints, we defer the proof showing that \autoref{cor:greedy_continuous_derivative} follows from \autoref{obs:greedy_discrete_derivative} to Appendix~\ref{app:no_guess_missing})}. To state the observation, let us define by $r$ and $r'$ the costliest and second costliest, respectively, elements in $OPT$ (recall that, by our assumption, $|OPT| \geq 2$). In other words, $c(r) \geq c(r') \geq \max_{v \in OPT \setminus\{r, r'\}} c(v)$.
\begin{observation} \label{obs:greedy_discrete_derivative}
For every $1 \leq i \leq \ell$,
\begin{align}
	\Delta_g(c(S_i)) \geq{} & 1 - g(c(S_i)) &\qquad& \text{if $c(S_i) \leq 1-c(r)$} \enspace,           \label{e:1} \\
	\Delta_g(c(S_i)) \geq{} & \f{1-g(c(S_i))-f(r)}{1-c(r)} && \text{if $c(S_i) \leq 1-c(r')$} \enspace,           \label{e:2} \\
	\Delta_g(c(S_i)) \geq{} & \f{1-g(c(S_i))-f(r)-f(r')}{1-c(r)-c(r')} && \text{if $c(S_i) \leq c(r)+c(r') < 1$} \enspace. \label{e:3}
\end{align}
\end{observation}
\begin{proof}
Inequalities~\eqref{e:1} and~\eqref{e:2} follow immediately from Lemma~\ref{l:r} by setting $S^* = OPT$ and $R = \varnothing$ or $R = \{r\}$, respectively. Inequality~\eqref{e:3} also follows from Lemma~\ref{l:r} by setting $S^* = OPT$ and $R = \{r, r'\}$ because the inequality $c(S_i) \leq c(r)+c(r')$ implies
\[
	\textstyle
	1 - c(S_i)
	\geq
	1 - c(r) - c(r')
	\geq
	c(OPT \setminus R)
	\geq
	\max_{v \in OPT \setminus R} c(v)
	\geq
	\max_{v \in OPT \setminus (R \cup S_i)} c(v)
	\enspace.
	\qedhere
\]
\end{proof}

\begin{restatable}{corollary}{corGreedyContinuousDerviative} \label{cor:greedy_continuous_derivative}
For every $t \in [0, 1]$ for which the derivative $g'(t)$ exists,
\[
	g'(t)
	\geq
	\max
	\begin{cases}
		1 - g(t) & \text{if $t \leq 1-c(r)$} \enspace, \\
		\f{1-g(t)-f(r)}{1-c(r)} & \text{if $t \leq 1-c(r')$} \enspace, \\
		\f{1-g(t)-f(r)-f(r')}{1-c(r)-c(r')} & \text{if $t \leq c(r)+c(r') < 1$} \enspace,\\
		0 & \text{always} \enspace,
	\end{cases}
\]
where one should understand the expression on the right hand side of the inequality as the maximum over the expressions corresponding to all the lines whose conditions hold.
\end{restatable}
\newcommand{\proofCorGreedyContinuousDerivative}{\begin{proof}
Consider first a value $t \in [0, c(S_\ell)]$ for which the derivative $g'(t)$ exists. For such a value the corollary follows from Observation~\ref{obs:greedy_discrete_derivative} because of the following three properties. First, as mentioned in Section~\ref{s:1_guess}, we have $g'(t) = \Delta_g(S_{i(t)})$; second, $g(c(S_{i(t)})) \leq g(t)$ because $g$ is a non-decreasing function and $c(S_{i(t)}) \leq t$; and finally, $g(t)$ is a non-decreasing function.

It remains to consider the case of $t > c(S_\ell)$. In this case the fact that {\Gr} terminated upon reaching the solution $S_\ell = S_{i(t)}$ implies that $S_\ell$ includes every element of $OPT$ of size at most $1 - t$. Therefore, by the monotonicity and submodularity of $f$,
\begin{align*}
	g(t)
	\geq
	f(\{v \in OPT \mid c(v) \leq 1 - t\})
	\geq{} &
	f(OPT) - f(\{v \in OPT \mid c(v) > 1- t\})\\
	={} &
	1 - f(\{v \in OPT \mid c(v) > 1- t\})
	\enspace.
\end{align*}
One can verify that the last inequality implies that the right hand side of the inequality in the lemma is always $0$ for $t > c(S_\ell)$ (note that every element $v \in OPT \setminus \{r, r'\}$ obeys $c(v) \leq c(OPT \setminus \{r, r'\}) \leq 1 - c(r) - c(r')$). Additionally, the left hand side of this inequality is $0$ by the definition of the continuous version of $g$, and thus, the inequality holds.
\end{proof}}
\inFull{\proofCorGreedyContinuousDerivative}

To get a guarantee for {\sc Greedy}, we need to get a lower bound on $f(S_\ell)=g(1)$. 
Theoretically, such a lower bound can be analytically proved 
by combining the lower bounds on $g'(t)$ proved by the last corollary.
However, to avoid tedious and non-insightful calculations, 
we use a computer to derive our lower bound.
\inFull{\inConference{
\section{Completing the Lower Bound for \texorpdfstring{{\Gr}}{Greedy}'s Approximation Ratio} \label{app:no_guess_missing}
In Section~\ref{ssc:lower_bound_greedy} we began the proof of the lower bound on the approximation guarantee of {\Gr} given in \autoref{t:3}. In this section we complete this proof. We begin with the proof of \autoref{cor:greedy_continuous_derivative}.

\corGreedyContinuousDerviative*
\proofCorGreedyContinuousDerivative

As explain in Section~\ref{ssc:lower_bound_greedy}, we would like to combine the lower bounds on $g'(t)$ proved by the last corollary to get a lower bound on $f(S_\ell)=g(1)$ (and thus, also on the guarantee of {\sc Greedy}), and our plan is to do that using a computer-based proof.
}%
Towards this goal, assume that we want to prove that the approximation ratio of {\Gr} 
is at least some target value $\rho \in (0, 1/2)$.
If $f(r) \geq \rho$ or $f(r') \geq \rho$, 
then this is trivial since {\sc Greedy} outputs a set that is better than any feasible singleton set 
(both $\{r\}$ and $\{r'\}$ are feasible solutions since $r,r' \in OPT$).
Therefore, the interesting case is when $f(r) < \rho$ and $f(r') < \rho$.

The last two inequalities mean that our computer based proof does not need to take into account the specific values of the elements $r$ and $r'$. However, we do not have such a nice property for the costs of these elements. The na\"{i}ve approach for handling this hurdle is to consider all the possible values for these costs, but this cannot be done since there are infinitely many such values. As an alternative, we develop below a way to lower bound $g(1)$ using a computer program given approximate values for $c(r)$ and $c(r')$. We later use the program to lower bound $g(1)$ for a large enough set of possible estimates so that every possible choice of real costs is close enough to one of the considered estimates, which makes the analysis apply to it.

We denote the estimates of $c(r)$ and $c(r')$ by $\tilde{c}(r)$ and $\tilde{c}(r')$, respectively. Lemma~\ref{lem:recursive_lower_bound} shows that, if these estimates are close enough to $c(r)$ and $c(r')$, then one can lower bound $g(1)$ using a computer program that (roughly) calculates the recursive series $m(i)$ defined as follows. Given any value $\delta \in (0, 1)$, and assuming $0 \leq \tilde{c}(r) + \tilde{c}(r') < 1$, we define $m(0) = 0$, and for every integer $1 \leq i$ we define
\begin{equation} \label{eq:recursive_series}
m(i)=
\max\begin{cases}
		\frac{m(i - 1) + \delta}{1 + \delta} & \text{if $i \leq \delta^{-1}(1 - \tilde{c}(r)) - 1$} \enspace,\\
		\frac{(1 - \tilde{c}(r)) \cdot m(i - 1) + \delta(1 - \rho)}{1 - \tilde{c}(r) + \delta} & \text{if $i \leq \delta^{-1}(1 - \tilde{c}(r')) - 1$} \enspace,\\
		\frac{(1 - \tilde{c}(r) - \tilde{c}(r')) \cdot m(i - 1) + \delta(1 - 2\rho)}{1 - \tilde{c}(r) - \tilde{c}(r') + \delta} & \text{if $i \leq \delta^{-1}(\tilde{c}(r) + \tilde{c}(r'))$} \enspace,\\
		m(i - 1) & \text{always} \enspace.
	\end{cases}
\end{equation}

\begin{lemma} \label{lem:recursive_lower_bound}
Given the assumptions,
\begin{align*}
	&\text{(i)} &&\mspace{-36mu}\max\{f(r), f(r')\} < \rho \enspace, &\qquad &
	\text{(ii)} &&\mspace{-36mu}\text{$\delta^{-1}$ is an integer} \enspace,\\
	&\text{(iii)} &&\mspace{-36mu}\tilde{c}(r) \in [c(r) - \delta, c(r)] \enspace, &&
	\text{(iv)} &&\mspace{-36mu}\tilde{c}(r') \in [c(r') - \delta, c(r')] \enspace,\\
	&\text{(v)}&& \mspace{-36mu} \tilde{c}(r) + \tilde{c}(r') < 1 \enspace,
\end{align*}
we have $g(i\delta) \geq m(i)$ for every $0 \leq i \leq \delta^{-1}$; and therefore, $g(1) = g(\delta^{-1} \cdot \delta) \geq m(\delta^{-1})$.
\end{lemma}
\begin{proof}
We prove the lemma by induction on $i$. For $i = 0$ the lemma holds because the non-negativity of $f$ implies 
$g(0) = f(S_0) = f(\varnothing) \geq 0 = m(0)$. Assume now that the lemma holds for $i - 1$ (for some integer $1 \leq i \leq \delta^{-1}$), 
and let us prove it for $i$. By plugging the assumptions of the lemma and the monotonicity of $g$ into the guarantee of Corollary~\ref{cor:greedy_continuous_derivative}, we get for every $\tau \in [\delta(i - 1), \delta i]$ for which $g'(\tau)$ is defined that
\begin{align*}
	g'(\tau)
	\geq{} &
	\max
	\begin{cases}
		1 - g(\tau) & \text{if $\tau \leq 1-c(r)$} \enspace, \\
		\f{1-g(\tau)-f(r)}{1-c(r)} & \text{if $\tau \leq 1-c(r')$} \enspace, \\
		\f{1-g(\tau)-f(r)-f(r')}{1-c(r)-c(r')} & \text{if $\tau \leq c(r)+c(r') < 1$} \enspace,\\
		0 & \text{always}
	\end{cases}\\
	\geq{} &
	\max
	\begin{cases}
		1 - g(\delta i) & \text{if $\tau \leq 1-\tilde{c}(r) - \delta$} \enspace, \\
		\f{1-g(\delta i) - \rho}{1-\tilde{c}(r)} & \text{if $\tau \leq 1-{c}(r') - \delta$} \enspace, \\
		\f{1-g(\delta i) - 2\rho}{1-\tilde{c}(r)-\tilde{c}(r')} & \text{if $\tau \leq \tilde{c}(r)+\tilde{c}(r')$} \enspace,\\
		0 & \text{always} \enspace.
	\end{cases}
\end{align*}
Note that in the second inequality we have dropped the condition $c(r)+c(r') < 1$. To see that this drop is of no consequence, we need to show that the assumptions of the lemma do not allow the equality $c(r) + c(r') = 1$ to hold. If this equality holds, then we must have $OPT = \{r, r'\}$. However, since $\rho \in (0, 1/2)$, this implies $f(r) + f(r') \geq f(OPT) = 1 > 2\rho$, which violates assumption (i) of the lemma.

Let us denote by $\Delta$ the difference $g(\delta i) - g(\delta(i - 1))$. Since $g$ is a continuous function whose derivative is defined for almost every $t \in [\delta(i - 1), \delta i]$, the last inequality implies
\begin{align*}
	\Delta
	\geq{} &
	\delta \cdot \max
	\begin{cases}
		1 - g(\delta i) & \text{if $\delta i \leq 1-\tilde{c}(r) - \delta$} \enspace, \\
		\f{1-g(\delta i) - \rho}{1-\tilde{c}(r)} & \text{if $\delta i \leq 1-{c}(r') - \delta$} \enspace, \\
		\f{1-g(\delta i) - 2\rho}{1-\tilde{c}(r)-\tilde{c}(r')} & \text{if $\delta i \leq \tilde{c}(r)+\tilde{c}(r')$} \enspace,\\
		0 & \text{always}
	\end{cases}\\
	={} &
	\delta \cdot \max
	\begin{cases}
		1 - g(\delta (i - 1)) - \Delta & \text{if $\delta i \leq 1-\tilde{c}(r) - \delta$} \enspace, \\
		\f{1-g(\delta (i - 1)) - \Delta - \rho}{1-\tilde{c}(r)} & \text{if $\delta i \leq 1-{c}(r') - \delta$} \enspace, \\
		\f{1-g(\delta (i - 1)) - \Delta - 2\rho}{1-\tilde{c}(r)-\tilde{c}(r')} & \text{if $\delta i \leq \tilde{c}(r)+\tilde{c}(r')$} \enspace,\\
		0 & \text{always} \enspace.
	\end{cases}
\end{align*}

Observe now that the last inequality is equivalent to $4$ inequalities, one corresponding to each line of the $\max$ operation. Isolating $\Delta$ in each one of theses inequalities yields
\[
	\Delta
	\geq
	\delta \cdot \max
	\begin{cases}
		\frac{1 - g(\delta (i - 1))}{1 + \delta} & \text{if $\delta i \leq 1-\tilde{c}(r) - \delta$} \enspace, \\
		\f{1-g(\delta (i - 1)) - \rho}{1-\tilde{c}(r) + \delta} & \text{if $\delta i \leq 1-{c}(r') - \delta$} \enspace, \\
		\f{1-g(\delta (i - 1)) - 2\rho}{1-\tilde{c}(r)-\tilde{c}(r') + \delta} & \text{if $\delta i \leq \tilde{c}(r)+\tilde{c}(r')$} \enspace,\\
		0 & \text{always} \enspace.
	\end{cases}
\]
To complete the proof of the lemma, it remains to observe that by adding $g(\delta(i - 1))$ to both sides of the last inequality we get
{\allowdisplaybreaks
\begin{align*}
	g(\delta i)
	\geq{} &
	\max
	\begin{cases}
		\frac{g(\delta (i - 1)) + \delta}{1 + \delta} & \text{if $\delta i \leq 1-\tilde{c}(r) - \delta$} \enspace, \\
		\f{(1 - \tilde{c}(r)) \cdot g(\delta (i - 1)) + \delta(1 - \rho)}{1-\tilde{c}(r) + \delta} & \text{if $\delta i \leq 1-{c}(r') - \delta$} \enspace, \\
		\f{(1 - \tilde{c}(r) - \tilde{c}(r')) \cdot g(\delta (i - 1)) + \delta(1 - 2\rho)}{1-\tilde{c}(r)-\tilde{c}(r') + \delta} & \text{if $\delta i \leq \tilde{c}(r)+\tilde{c}(r')$} \enspace,\\
		0 & \text{always}
	\end{cases}\\
	\geq{} &
	\begin{cases}
		\frac{m(i - 1) + \delta}{1 + \delta} & \text{if $\delta i \leq 1-\tilde{c}(r) - \delta$} \enspace, \\
		\f{(1 - \tilde{c}(r)) \cdot m(i - 1) + \delta(1 - \rho)}{1-\tilde{c}(r) + \delta} & \text{if $\delta i \leq 1-{c}(r') - \delta$} \enspace, \\
		\f{(1 - \tilde{c}(r) - \tilde{c}(r')) \cdot m(i - 1) + \delta(1 - 2\rho)}{1-\tilde{c}(r)-\tilde{c}(r') + \delta} & \text{if $\delta i \leq \tilde{c}(r)+\tilde{c}(r')$} \enspace,\\
		0 & \text{always}
	\end{cases}
	=
	m(i)
	\enspace,
\end{align*}%
}
where the inequality holds by the induction hypothesis.
\end{proof}

\begin{corollary} \label{cor:guarantee_with_good_estimates}
Given any value $\rho$, any value $\delta \in (0, 1)$ such that $\delta^{-1}$ is an integer and values for $\tilde{c}(r)$ and $\tilde{c}(r')$ that obey conditions (iii)-(v) of Lemma~\ref{lem:recursive_lower_bound}, the approximation ratio of {\Gr} is at least $\min\{m(\delta^{-1}),\rho\}$.
\end{corollary}
\begin{proof}
We consider two cases in this proof. The first case is when condition (i) of Lemma~\ref{lem:recursive_lower_bound} is violated, i.e., $\max\{f(r), f(r')\} \geq \rho$. Recall now that $\{r\}$ and $\{r'\}$ are both feasible solutions because $r, r' \in OPT$, and {\Gr} considers all the feasible singletons as possible solutions. These facts imply together that in this case the value of the output of {\Gr} is at least $\max\{f(r), f(r')\} \geq \rho$.

Consider now the case in which condition (i) of Lemma~\ref{lem:recursive_lower_bound} holds. Since the condition (ii) of this lemma holds by the properties of $\delta$ and we assumed that the other conditions of the lemma hold as well, we get by Lemma~\ref{lem:recursive_lower_bound} that
\[
	f(S_\ell)
	=
	g(1)
	\geq
	m(\delta^{-1})
	\enspace.
\]
As $S_\ell$ is one of the candidate solutions considered by {\Gr}, the value of the output of {\Gr} in this case is at least $m(\delta^{-1})$.
\end{proof}

If we want to use Corollary~\ref{cor:guarantee_with_good_estimates} to lower bound the approximation ratio of {\Gr}, then we need to choose values for the four parameters $\rho$, $\delta$, $\tilde{c}(r)$ and $\tilde{c}(r')$ that obey all the requirements of the corollary. It is not difficult to do so for $\rho$ and $\delta$, but there is no possible assignment of values for $\tilde{c}(r)$ and $\tilde{c}(r')$ that will be good for all instances of {\BSM} (because the conditions of \autoref{lem:recursive_lower_bound} require $\tilde{c}(r)$ and $\tilde{c}(r')$ to be close to $c(r)$ and $c(r')$, respectively). As explained above, we solve this issue by considering a set $\cC(\delta)$ of possible pairs of values for $\tilde{c}(r)$ and $\tilde{c}(r')$ that is large enough so that it always includes at least one pair of good values, and then applying \autoref{cor:guarantee_with_good_estimates} independently to every pair from $\cC(\delta)$. The next proposition formally states the guarantee that we get in this way. To state this proposition, we define the set $\cC(\delta)$ as follows.
\[
	\cC(\delta) = \{(j \delta, j' \delta) \mid \text{$j$ and $j'$ are non-negative integers obeying $j \leq j'$ and $j + j' < \delta^{-1}$}\}
	\enspace.
\]
We also need to recall that, despite their omission from the notation we have used so far, the values of $\rho$, $\tilde{c}(r)$ and $\tilde{c}(r')$ also affect $m(\delta^{-1})$. To make this more explicit, we use in this proposition the expression $m(\rho, \tilde{c}(r), \tilde{c}(r'), \delta^{-1})$ to denote the value of $m(\delta^{-1})$ corresponding to a particular choice of values for these parameters.
\begin{proposition} \label{prop:symbolic_lower_bound}
Given any value $\rho$, any value $\delta \in (0, 1)$ such that $\delta^{-1}$ is an integer, the approximation ratio of {\Gr} is at least
\begin{equation} \label{eq:lower_bound}
	\min\left\{\rho, \min_{(\tilde{c}(r), \tilde{c}(r')) \in \cC(\delta)} m(\rho, \tilde{c}(r), \tilde{c}(r'), \delta^{-1})\right\}
	\enspace.
\end{equation}
\end{proposition}
\begin{proof}
We note that the proposition follows immediately from Corollary~\ref{cor:guarantee_with_good_estimates} if we are guaranteed that the set $\cC(\delta)$ includes a pair $(\tilde{c}(r), \tilde{c}(r'))$ that obeys conditions (iii)-(v) of Lemma~\ref{lem:recursive_lower_bound}. One can also note that the definition of $\cC(\delta)$ guarantees that every pair $(\tilde{c}(r), \tilde{c}(r')) \in \cC(\delta)$ obeys condition (v) of the lemma. Thus, it remains to prove that there exists a pair $(\tilde{c}(r), \tilde{c}(r')) \in \cC(\delta)$ obeying conditions (iii) and (iv) of Lemma~\ref{lem:recursive_lower_bound}, which is our objective in the rest of this proof.

Let $\tilde{c}(r) = \delta \max\{0, \lceil c(r) / \delta \rceil - 1\}$ and $\tilde{c}(r') = \delta \max\{0, \lceil c(r') / \delta \rceil - 1\}$. Clearly both maximums in these definitions are non-negative integer numbers because $c(r)$ and $c(r')$ are positive numbers, and the first maximum is at least as large as the second one because $c(r) \geq c(r')$ by the definitions of $r$ and $r'$. Furthermore, the two maximums are strictly smaller than $c(r) / \delta$ and $c(r') / \delta$, respectively, and therefore, their sum is smaller than $[c(r) + c(r')]/\delta \leq \delta^{-1}$. Hence, the pair $(\tilde{c}(r), \tilde{c}(r'))$ we have defined belongs to $\cC(\delta)$. We can also observe that $\tilde{c}(r)$ obeys condition (iii) of Lemma~\ref{lem:recursive_lower_bound} (i.e., $\tilde{c}(r) \in [c(r) - \delta, c(r)]$) because $\lceil c(r) / \delta \rceil - 1 \in [c(r)/\delta - 1, c(r)/\delta]$. A similar arguments shows that $\tilde{c}(r')$ obeys condition (iv) of the same lemma, and thus, completes the proof.
\end{proof}

\inFull{Appendix}\inConference{Section}~\ref{app:code} gives VB.net code that numerically shows that $0.427$ is a lower bound on $\min\left\{\rho, \min_{(\tilde{c}(r), \tilde{c}(r')) \in \cC(\delta)} m(\rho, \tilde{c}(r), \tilde{c}(r'), \delta^{-1})\right\}$ for $\rho=0.427$ and $\delta = 10^{-3}$. By the last proposition, this implies the lower bound on the approximation ratio of {\Gr} stated in Theorem~\ref{t:3}. Our code crucially relies on the fact that if one plugs a lower bound on $m(i - 1)$ into the right hand side of Inequality~\eqref{eq:recursive_series}, then one gets a lower bound on $m(i)$.

\inConference{\inFull{\section{Code}}
\inConference{\subsection{Code}}
\label{app:code}

This \inFull{appendix}\inConference{section} includes the code used to get a lower bound on the expression
\[
	\min\left\{\rho, \min_{(\tilde{c}(r), \tilde{c}(r')) \in \cC(\delta)} m(\rho, \tilde{c}(r), \tilde{c}(r'), \delta^{-1})\right\}
\]	
mentioned in Proposition~\ref{prop:symbolic_lower_bound}. This code consists of two parts. The first part, given in Section~\ref{sec:structure}, describes a structure used to represent non-negative rational numbers. The second part, given in Section~\ref{sec:main}, describes the main program, which uses the structure defined by the first part.

\inFull{\subsection{Structure Representing a Non-negative Rational Number}}
\inConference{\subsubsection{Structure Representing a Non-negative Rational Number}}
\label{sec:structure}

In this section we describe an immutable structure used to represent a non-negative rational number. The public member functions of this structure support various operations on such numbers. Some of these operations return the output one would expect, while others return a lower bound on this output. Member functions of the last kind are denoted by the prefix LB.

\raggedbottom
\inFull{
\begin{lstlisting}
Public Structure RationalNumber

    ' A mask having ones in the low 32 bits and zeros in the higher bits
    Private Const UIntegerMask As ULong = &HFFFFFFFFL

    ' A mask having a one only in the most significant bit
    Private Const MSBMask As ULong = (CULng(1) << 63)

    ' The numerator and denominator of the rational number. The invariant
    ' of the structure is that the denominator is non-zero.
    Private ReadOnly Numerator As UInteger
    Private ReadOnly Denominator As UInteger

    ' This constructor initialzes a rational number equal to a given
    ' integal number.
    Public Sub New(IntergalNumber As UInteger)
        Me.Numerator = IntergalNumber
        Me.Denominator = 1
    End Sub

    ' This constructor initializes a rational number by specifying the
    ' numerator and denominator.
    Public Sub New(Numerator As UInteger, Denominator As UInteger)
        If Denominator = 0 Then Throw New DivideByZeroException(
            "The denominator of a rational number cannot be 0.")
        Me.Numerator = Numerator
        Me.Denominator = Denominator
    End Sub

    ' Calculates the GCD of two numbers, at least one of which is non-zero,
    ' using Euclid's algorithm.
    Private Shared Function GCD(a As ULong, b As ULong) As ULong
        If b = 0 Then
            Return a
        ElseIf b > a Then
            Return GCD(b, a)
        Else
            Return GCD(b, a Mod b)
        End If
    End Function

    ' This function returns the 0 based index of the most significant one bit
    ' in the given number.
    Private Shared Function GetMSBIndex(Number As ULong) As Integer
        Dim Output As Integer = 0
        Do While Number > 0
            Output += 1
            Number >>= 1
        Loop
        Return Output
    End Function

    ' This function returns a lower bound on the product of a rational number
    ' and an integer.
    Public Shared Function LBMultiplication(Number1 As RationalNumber,
                                 Number2 As UInteger) As RationalNumber
        Return LBMultiplication(Number1, New RationalNumber(Number2))
    End Function


    ' This function returns a lower bound on the product of two rational numbers.
    Public Shared Function LBMultiplication(Number1 As RationalNumber,
                                  Number2 As RationalNumber) As RationalNumber
        ' Calculate the exact multiplication
        Dim LongNumerator As ULong = CULng(Number1.Numerator) _
                                        * CULng(Number2.Numerator)
        Dim LongDenominator As ULong = CULng(Number1.Denominator) _
                                        * CULng(Number2.Denominator)

        ' Reduce the fraction as much as possible
        Dim GCDValue As ULong = GCD(LongNumerator, LongDenominator)
        LongNumerator = LongNumerator \ GCDValue
        LongDenominator = LongDenominator \ GCDValue

        ' Force the numerator and denominator into 32 bits (even if this results
        ' in lossing some value). The variable DecreaseBy contains the number of
        ' bits that should be removed.
        Dim DecreaseBy As Integer = Math.Max(GetMSBIndex(Math.Max(LongNumerator,
                                                       LongDenominator)) - 31, 0)
        Dim Numerator As UInteger = CUInt(LongNumerator >> DecreaseBy)
        Dim Denominator As UInteger = CUInt(LongDenominator >> DecreaseBy)
        If (DecreaseBy > 0) And (LongDenominator << (64 - DecreaseBy) > 0) Then
            ' We now need to increase the denominator by 1 to compensate
            ' for the bits lost. However, in a rare case this might result
            ' in an overflow, and thus, might require the removal of one more
            ' bit from the precision.
            If Denominator = UInteger.MaxValue Then
                Numerator >>= 1
                Denominator >>= 1
            End If
            Denominator += CUInt(1)
        End If

        ' Return the rational number obtained
        Return New RationalNumber(Numerator, Denominator)
    End Function

    ' This function returns a lower bound on the result of the division of
    ' a given rational number by a given integer.
    Public Shared Function LBDivision(Dividend As RationalNumber,
                                      Divisor As UInteger) As RationalNumber
        Return LBMultiplication(Dividend, New RationalNumber(1, Divisor))
    End Function

    ' This function returns a lower bound on the sum of two rational numbers.
    Public Shared Function LBAddition(Number1 As RationalNumber,
                                      Number2 As RationalNumber) As RationalNumber
        ' Make the denominator common
        Dim DenominatorsGCD As ULong = GCD(Number1.Denominator,
                                           Number2.Denominator)
        Dim LongNumerator1 As ULong = CULng(Number1.Numerator) _
                                    * CULng(Number2.Denominator) \ DenominatorsGCD
        Dim LongNumerator2 As ULong = CULng(Number2.Numerator) _
                                    * CULng(Number1.Denominator) \ DenominatorsGCD
        Dim LongDenominator As ULong = CULng(Number1.Denominator) _
                                    * CULng(Number2.Denominator) \ DenominatorsGCD

        ' Reduce precision by 1 bit if the most significant bit of either
        ' numerator is 1.
        If ((LongNumerator1 And MSBMask) > 0) Or
                                           ((LongNumerator2 And MSBMask) > 0) Then
            LongNumerator1 >>= 1
            LongNumerator2 >>= 1
            If (LongDenominator And CULng(1)) = 0 Then
                LongDenominator >>= 1
            Else
                LongDenominator = (LongDenominator >> 1) + CUInt(1)
            End If
        End If

        ' Do the addition itself
        Dim LongResultNumerator As ULong = LongNumerator1 + LongNumerator2

        ' Reduce the fraction as much as possible
        Dim ReductionGCD As ULong = GCD(LongResultNumerator, LongDenominator)
        LongResultNumerator = LongResultNumerator \ ReductionGCD
        LongDenominator = LongDenominator \ ReductionGCD

        ' Force the numerator and denominator into 32 bits (even if this results
        ' in lossing some value). The variable DecreaseBy contains the number of
        ' bits that should be removed.
        Dim DecreaseBy As Integer = Math.Max(GetMSBIndex(Math.Max(
                                   LongResultNumerator, LongDenominator)) - 31, 0)
        Dim ResultNumerator As UInteger = CUInt(LongResultNumerator >> DecreaseBy)
        Dim Denominator As UInteger = CUInt(LongDenominator >> DecreaseBy)
        If (DecreaseBy > 0) And (LongDenominator << (64 - DecreaseBy) > 0) Then
            ' We now need to increase the denominator by 1 to compensate
            ' for the bits lost. However, in a rare case this might result
            ' in an overflow, and thus, might require the removal of one more
            ' bit from the precision.
            If Denominator = UInteger.MaxValue Then
                ResultNumerator >>= 1
                Denominator >>= 1
            End If
            Denominator += CUInt(1)
        End If

        ' Return the rational number obtained
        Return New RationalNumber(ResultNumerator, Denominator)
    End Function

    ' This fucntion gets two rational numbers and returns their maximum.
    Public Shared Function Max(Number1 As RationalNumber,
                               Number2 As RationalNumber) As RationalNumber
        Dim Numerator1TimesDenominator2 As ULong = CULng(Number1.Numerator) _
                                                    * CULng(Number2.Denominator)
        Dim Numerator2TimesDenominator1 As ULong = CULng(Number1.Denominator) _
                                                    * CULng(Number2.Numerator)
        If Numerator1TimesDenominator2 > Numerator2TimesDenominator1 Then
            Return Number1
        Else
            Return Number2
        End If
    End Function

    ' This fucntion gets two rational numbers and returns their minimum.
    Public Shared Function Min(Number1 As RationalNumber,
                               Number2 As RationalNumber) As RationalNumber
        Dim Numerator1TimesDenominator2 As ULong = CULng(Number1.Numerator) _
                                                    * CULng(Number2.Denominator)
        Dim Numerator2TimesDenominator1 As ULong = CULng(Number1.Denominator) _
                                                    * CULng(Number2.Numerator)
        If Numerator1TimesDenominator2 < Numerator2TimesDenominator1 Then
            Return Number1
        Else
            Return Number2
        End If
    End Function

    ' This function returns a string representation of a lower bound of the
    ' number represented by the structure.
    Public Overrides Function ToString() As String
        ' Initialize the output to contain the integral part of the output
        Dim Output As New System.Text.StringBuilder(
                                        (Numerator \ Denominator).ToString)

        ' Produce a decimal dot if there is a remainder
        Dim Remainder As ULong = Numerator Mod Denominator
        If Remainder > 0 Then Output.Append("."c)

        ' Produce up to 7 digits of the fractional part
        For i As Integer = 1 To 7
            If Remainder = 0 Then Exit For
            Remainder *= CULng(10)
            Output.Append(Remainder \ Denominator)
            Remainder = Remainder Mod Denominator
        Next i

        ' Return the constructed output
        Return Output.ToString
    End Function

End Structure
\end{lstlisting}}
\flushbottom

\inFull{\subsection{Main Program}}
\inConference{\subsubsection{Main Program}}
\label{sec:main}

In this section we give the main program used to evaluate the expression given in Proposition~\ref{prop:symbolic_lower_bound}. This program uses the structure defined in Section~\ref{sec:structure}.

\raggedbottom
\begin{lstlisting}
Module MainProgram

    ' The value of \delta^{-1}.
    Private Const OneOverDelta As UInteger = 1000

    ' The value of \rho / \delta
    Private Const RhoOverDelta As UInteger = 427

    ' This function gets \tilde{c}(r)/\delta and \tilde{c}(r')/\delta, and then 
    ' returns a lower bound for m(\delta^{-1}).
    Private Function CalculateMOneOverDelta(CrOverDelta As UInteger,
                                 CrPrimeOverDelta As UInteger) As RationalNumber
        ' The following three rational values are used often in the calculations
        ' below.
        Static One As New RationalNumber(1)
        Static OneMinusRho As New RationalNumber(OneOverDelta - RhoOverDelta,
                                                 OneOverDelta)
        Static OneMinusTwoRho As New RationalNumber(OneOverDelta -
                                                   2 * RhoOverDelta, OneOverDelta)

        ' A lower bound on the value of m(i) corresponding to the current i
        Dim m As New RationalNumber(0)

        ' Calculate a lower bound for m(i) for each i in 1 to \delta^{-1}
        For i As Integer = 1 To OneOverDelta
            ' Calculate the bound in the first row in the definition of m(i). Note
            ' that (1) compared to the paper, we devide both the numerator and
            ' denominator by \delta, and (2) if the condition of the row
            ' evaluates to false, the variable FirstBound simply gets the bound
            ' from the fourth row of the definition.
            Dim FirstBound As RationalNumber = m
            If i <= OneOverDelta - CrOverDelta - 1 Then
                FirstBound = RationalNumber.LBDivision(RationalNumber.LBAddition(
                   RationalNumber.LBMultiplication(m, OneOverDelta), One),
                   OneOverDelta + CUInt(1))
            End If

            ' Calculate the bound in the second row in the definition of m(i).
            ' The two remarks regarding the first row apply to this calculation 
            ' as well.
            Dim SecondBound As RationalNumber = m
            If i <= OneOverDelta - CrPrimeOverDelta - 1 Then
                SecondBound = RationalNumber.LBDivision(RationalNumber.LBAddition(
                   RationalNumber.LBMultiplication(m, OneOverDelta - CrOverDelta),
                   OneMinusRho), OneOverDelta - CrOverDelta + CUInt(1))
            End If

            ' Calculate the bound in the third row in the definition of m(i).
            ' The two remarks regarding the first row apply to this calculation
            ' as well.
            Dim ThirdBound As RationalNumber = m
            If i <= CrOverDelta + CrPrimeOverDelta Then
                ThirdBound = RationalNumber.LBDivision(RationalNumber.LBAddition(
                   RationalNumber.LBMultiplication(m, OneOverDelta _
                    - CrOverDelta - CrPrimeOverDelta), OneMinusTwoRho),
                   OneOverDelta - CrOverDelta - CrPrimeOverDelta + CUInt(1))
            End If

            ' The lower bound on the value of m(i) is the maximum between the
            ' three above bounds and the lower bound on m(i - 1)
            m = RationalNumber.Max(m, RationalNumber.Max(FirstBound,
                                     RationalNumber.Max(SecondBound, ThirdBound)))
        Next i

        Return m
    End Function

    Sub Main()
        ' The minimum lower bound for m(i) obtained for any choice of \tilde{c}(r)
        ' and \tilde{c}(r')
        Dim MinimumM As RationalNumber = New RationalNumber(1)

        ' Iterate over all the pairs of values for \tilde{c}(r) and \tilde{c}(r') 
        ' in \mathcal{C}, calculate m(\delta^{-1}) for each pair and update
        ' MinimumM accordingly.
        For CrOverDelta As UInteger = 0 To OneOverDelta - 1
            For CrPrimeOverDelta As UInteger = 0 To Math.Min(CrOverDelta,
                                            OneOverDelta - CrOverDelta - CUInt(1))
                Dim LowerBoundOnMOneOverDelta As RationalNumber =
                    CalculateMOneOverDelta(CrOverDelta, CrPrimeOverDelta)
                MinimumM = RationalNumber.Min(MinimumM, LowerBoundOnMOneOverDelta)
            Next CrPrimeOverDelta
        Next CrOverDelta

        ' Output the guarantee we have got on the approximtion ratio of Greedy.
        Call Console.WriteLine("The lowest lower bound on m(\\delta) obtained " _
             + "for any guess of \\tilde{{c}}(r) And \\tilde{{c}}(r') was {0}.",
                               MinimumM)
        Call Console.WriteLine("Thus, the approximation ratio of Greedy is at " _
             + "least the minimum betweeen this value and \\rho, i.e., {0}.",
                                RationalNumber.Min(MinimumM,
                                New RationalNumber(RhoOverDelta, OneOverDelta)))
    End Sub

End Module
\end{lstlisting}
\flushbottom}}%
\inConference{The details are deferred to Appedix~\ref{app:no_guess_missing}.}

\subsection{Upper Bounding {\Gr}'s Approximation Ratio} \label{ssc:upper_bound_greedy}

In this section we prove the upper bound on the approximation ratio of {\sc Greedy} stated in \autoref{t:3}. We do that by describing an example instance on which {\Gr} performs relatively poorly. Our example instance is parametrized by three values: a large enough positive integer $n$, a small enough positive value $\eps$ and a value $\alpha \in (1/3, 0.49)$ to be determined later. Let us define two auxiliary sets $X = \{x_i \mid 1 \leq i \leq n\}$ and $Y =  \{y_i \mid 1 \leq i \leq n\}$. Then, the ground set of our example instance is \inConference{the set $V = \{z_1, z_2, w\} \cup X \cup Y$}\inFull{the following set}, which consists of $2n + 3$ elements.
\inFull{\[
	V
	=
	\{z_1, z_2, w\} \cup X \cup Y
	\enspace.
\]}%
The objective function of our example instance is given, for every set $S \subseteq V$, by
\begin{align*}
	f(S)
	=
	1 - &\left(1 - \frac{1 - \alpha}{n}\right)^{|S \cap X|}\left[\alpha\left(2 - \frac{|S \cap \{z_1, z_2\}|}{1 + 2\eps}\right)\right.\\
	&+ \left.(1 - 2\alpha)\left(1 - \frac{|S \cap \{w\}|}{1 + 2\eps}\right)\left(1 - \frac{\alpha(1 - 2\alpha)^{-1} - 1}{n}\right)^{|S \cap Y|}\right]
	\enspace.
\end{align*}
Finally, the budget $B$ of our example instance is $1$ (as usual), and its cost function $c$ is described by the following table.
\begin{center}\begin{tabular}{lcccc}
	\hline
	Element of $V$ & $z_i$	&	$w$ & $x_i$ & $y_i$ \\
	\hline
	Cost & $\alpha$ & $1 - 2\alpha$ & $\frac{(1 - \alpha)(1 + \eps)}{n}$  & $\frac{(3\alpha - 1)(1 + 1.5\eps)}{n}$\\
	\hline
\end{tabular}\end{center}

Intuitively, the optimal set of the example instance is the set $O = \{z_1, z_2, w\}$. However, the instance is constructed in such a way that {\Gr} chooses $X \cup Y$ as its solution. Specifically, {\Gr} begins by picking the elements of $X$. These elements have a slightly higher density than the elements of $O$, but taking them diminishes the marginal values of all unpicked elements. When the set $X$ is exhausted, the elements $z_1$ and $z_2$ can no longer be taken because their costs exceed the budget available at that point; thus, diminishing their marginal values further does not affect the value of the output of {\Gr} in any way. Accordingly, {\Gr} begins to pick at that point the elements of $Y$, and taking them only diminishes the marginal values of unpicked $Y$ elements and the marginal value of $w$. 
\inConference{One can verify that the ratio between the set $X \cup Y$ picked by {\Gr} according to this intuition and the optimal set $O$ is at most 0.462, and therefore, the upper bound in Theorem~\ref{t:3} indeed follows from our example instance.}

\inConference{Due to space constraints, we defer the formal analysis of our example instance to Appendix~\ref{app:example_instance_analysis}.}
\inFull{\inFull{We now get to the formal analysis of our example instance, which we begin by proving that its objective function has all the properties that it needs to have in order to make the example instance a legal instance of {\BSM}.
}

\inConference{
\section{Formal Analysis of the Example Instance from Section~\ref{ssc:upper_bound_greedy} \label{app:example_instance_analysis}}

In this section we give the formal analysis of our example instance from Section~\ref{ssc:upper_bound_greedy}. We begin by proving that the objective function of the example instance has all the properties that it needs to have in order to make the example instance a legal instance of {\BSM}.
}

\begin{lemma}
The objective function $f$ is non-negative, monotone and submodular.
\end{lemma}
\begin{proof}
Let us define a few additional set functions as follows. For every $S \subseteq V$,
\[\begin{array}{ll}
	f_1(S) = \left(1 - \frac{1 - \alpha}{n}\right)^{|S \cap X|}
	&
	f_2(S) = 2 - \frac{|S \cap \{z_1, z_2\}|}{1 + 2\eps}
	\\[1mm]
	f_3(S) = 1 - \frac{|S \cap \{w\}|}{1 + 2\eps}
	&
	f_4(S) = \left(1 - \frac{\alpha(1 - 2\alpha)^{-1} - 1}{n}\right)^{|S \cap Y|} \enspace.
\end{array}\]
Given these functions, one can verify that
\[
	f(S)
	=
	1 - f_1(S) \cdot [\alpha \cdot f_2(S) + (1 - 2\alpha) \cdot f_3(S) \cdot f_4(S)]
	\enspace.
\]
Thus, the non-negativity of $f$ follows from the observations that for every set $S \subseteq V$ we have
\[\begin{array}{ll}
	f_1(S) \in [\alpha, 1] \subseteq [0, 1]
	&
	f_2(S) \in [4\eps / (1 + 2\eps), 2] \subseteq [0, 2]
	\\[1mm]
	f_3(S) \in [2\eps / (1 + 2\eps), 1] \subseteq [0, 1]
	&
	f_4(S) \in [0, 1] \enspace,
\end{array}\]
where the last inclusion holds since the observations that $1 - 2\alpha \geq 1 - 2\cdot 0.49 = 0.02$ and $\alpha \geq 1/3 \geq 1 - 2\alpha$ imply together that $\left(1 - \frac{\alpha(1 - 2\alpha)^{-1} - 1}{n}\right) \in (0, 1]$ for a large enough $n$. Similarly, the monotonicity of $f$ follows from the observation that the functions $f_1$, $f_2$, $f_3$ and $f_4$ are all down-monotone (i.e., $-f_1$, $-f_2$, $-f_3$ and $-f_4$ are monotone functions).

It remains to prove that $f$ is submodular, which requires us to prove that for every element $u \in V$ the function $f(u \mid S - u)$ is a down-monotone function of $S$. We do that separately for each kind of element in the ground set of $V$. First, for an element $u \in \{z_1, z_2\}$, since the functions $f_1$, $f_3$ and $f_4$ ignore such elements, we get
\begin{align*}
	f(u \mid S&)
	=
	f(S + u) - f(S)
	=
	f_1(S) \cdot [\alpha \cdot f_2(S) + (1 - 2\alpha) \cdot f_3(S) \cdot f_4(S)] \\&\mspace{200mu}- f_1(S) \cdot [\alpha \cdot f_2(S + u) + (1 - 2\alpha) \cdot f_3(S) \cdot f_4(S)]\\
	={} &
	\alpha \cdot f_1(S) \cdot [f_2(S) - f_2(S + u)]
	=
	\frac{\alpha \cdot f_1(S)}{1 + 2\eps} \cdot [|(S + u) \cap \{z_1, z_2\}| - |S \cap \{z_1, z_2\}|]\\
	={} &
	\frac{\alpha \cdot f_1(S)}{1 + 2\eps} \cdot |\{u\} \setminus S|
	\enspace,
\end{align*}
which is clearly a down-monotone function of $S$ since $f_1$ is a non-negative down-monotone function of $S$.

Consider now the element $u = w$. Since the functions $f_1$, $f_2$ and $f_4$ ignore this element,
\begin{align*}
	f(w \mid S&)
	=
	f(S + w) - f(S)
	=
	f_1(S) \cdot [\alpha \cdot f_2(S) + (1 - 2\alpha) \cdot f_3(S) \cdot f_4(S)] \\&\mspace{200mu}- f_1(S) \cdot [\alpha \cdot f_2(S) + (1 - 2\alpha) \cdot f_3(S + w) \cdot f_4(S)]\\
	={} &
	(1 - 2\alpha) \cdot f_1(S) \cdot f_4(S) \cdot [f_3(S) - f_3(S + w)]\\
	={} &
	\frac{(1 - 2\alpha) \cdot f_1(S) \cdot f_3(S)}{1 + 2\eps} \cdot [|(S + w) \cap \{w\}| - |S \cap \{w\}|]\\
	={} &
	\frac{(1 - 2\alpha) \cdot f_1(S) \cdot f_3(S)}{1 + 2\eps} \cdot |\{w\} \setminus S|
	\enspace,
\end{align*}
which is clearly a down-monotone function of $S$ since $f_1$ and $f_3$ are non-negative down-monotone functions of $S$.

Next, consider an element $u \in X$. Since the functions $f_2$, $f_3$ and $f_4$ ignore such elements,
\begin{align*}
	f(u \mid S&)
	=
	f(S + u) - f(S)
	=
	f_1(S) \cdot [\alpha \cdot f_2(S) + (1 - 2\alpha) \cdot f_3(S) \cdot f_4(S)] \\&\mspace{200mu}- f_1(S + u) \cdot [\alpha \cdot f_2(S) + (1 - 2\alpha) \cdot f_3(S) \cdot f_4(S)]\\
	={} &
	[\alpha \cdot f_2(S) + (1 - 2\alpha) \cdot f_3(S) \cdot f_4(S)] \cdot [f_1(S) - f_1(S + u)]\\
	={} &
	[\alpha \cdot f_2(S) + (1 - 2\alpha) \cdot f_3(S) \cdot f_4(S)] \cdot f_1(S) \cdot \left[1 - \left(1 - \frac{1 - \alpha}{n}\right)^{|(S + u) \cap X| - |S \cap X|}\right]\\
	={} &
	[\alpha \cdot f_2(S) + (1 - 2\alpha) \cdot f_3(S) \cdot f_4(S)] \cdot f_1(S) \cdot \frac{(1 - \alpha)\cdot |\{u\} \setminus S|}{n}
	\enspace,
\end{align*}
which is clearly a down-monotone function of $S$ since $f_1$, $f_2$, $f_3$ and $f_4$ are non-negative down-monotone functions of $S$.

Finally, consider an element $u \in Y$. Since the functions $f_1$, $f_2$ and $f_3$ ignore such elements,
\begin{align*}
	f(u \mid S&)
	=
	f(S + u) - f(S)
	=
	f_1(S) \cdot [\alpha \cdot f_2(S) + (1 - 2\alpha) \cdot f_3(S) \cdot f_4(S)] \\&\mspace{200mu}- f_1(S) \cdot [\alpha \cdot f_2(S) + (1 - 2\alpha) \cdot f_3(S) \cdot f_4(S + u)]\\
	={} &
	(1 - 2\alpha) \cdot f_1(S) \cdot f_3(S) \cdot [f_4(S) - f_4(S + u)]\\
	={} &
	(1 - 2\alpha) \cdot f_1(S) \cdot f_3(S) \cdot f_4(S) \cdot \left[1 - \left(1 - \frac{\alpha(1 - 2\alpha)^{-1} - 1}{n}\right)^{|(S + u) \cap Y| - |S \cap Y|}\right]\\
	={} &
	(1 - 2\alpha) \cdot f_1(S) \cdot f_3(S) \cdot f_4(S) \cdot \frac{[\alpha(1 - 2\alpha)^{-1} - 1] \cdot |\{u\} \setminus S|}{n}
	\enspace,
\end{align*}
which is clearly a down-monotone function of $S$ since $f_1$, $f_3$ and $f_4$ are non-negative down-monotone functions of $S$ and $\alpha(1 - 2\alpha)^{-1} - 1 \geq 0$.
\end{proof}

Our next objective is to show that the set $O = \{z_1, z_2, w\}$ claimed in the above intuition to be the optimal solution is indeed a feasible solution with a high value (we do not need to formally prove that it is in fact optimal).
\begin{observation} \label{obs:O_properties}
$c(O) = 2\alpha + (1 - \alpha) = 1$, and thus, $O$ is a feasible solution. Furthermore, $f(O) = \frac{1}{1 + 2\eps}$.
\end{observation}
\begin{proof}
One can verify the first part of the observation by simply looking up the costs of the elements of $O$ in the definition of the example instance. To see that the second part of the observation also holds, we note that
\begin{align*}
	f(O) 
	={}&
	1 - \left(1 - \frac{1 - \alpha}{n}\right)^{|O \cap X|}\left[\alpha\left(2 - \frac{|O \cap \{z_1, z_2\}|}{1 + 2\eps}\right)\right.\\
	&\mspace{100mu}+ \left.(1 - 2\alpha)\left(1 - \frac{|O \cap \{w\}|}{1 + 2\eps}\right)\left(1 - \frac{\alpha(1 - 2\alpha)^{-1} - 1}{n}\right)^{|O \cap Y|}\right]\\
	={}&
	1 - \alpha\left(2 - \frac{|\{z_1, z_2\}|}{1 + 2\eps}\right) - (1 - 2\alpha)\left(1 - \frac{|\{w\}|}{1 + 2\eps}\right)\mspace{-400mu}\\
	={}&
	\frac{(1 + 2\eps) - \alpha[(2 + 4\eps) - 2] - (1 - 2\alpha)[(1 + 2\eps) - 1]}{1 + 2\eps}\\
	={} &
	\frac{1 + 2\eps - \alpha \cdot 4\eps - (1 - 2\alpha) \cdot 2\eps}{1 + 2\eps}
	=
	\frac{1}{1 + 2\eps}
	\enspace. \qedhere
\end{align*}
\end{proof}

To complete the analysis of the example instance, we need to show that {\Gr} outputs a solution of a much lower value than $O$ when the value of the parameter $\alpha$ is chosen appropriately. Since {\Gr} outputs the better among the best feasible singleton set and the output of {\PG}, we need to analyze these two options separately. The next lemma upper bounds the value of the best feasible singleton set.
\begin{lemma} \label{lem:singleton_value}
The value of the best singleton set is at most $\frac{\alpha}{1 + 2\eps}$, and thus, the same is true also for the value of the best \emph{feasible} singleton set.
\end{lemma}
\begin{proof}
We need to show that $f(\{u\}) \leq \alpha / (1 + 2\eps)$ for every element $u \in V$. We do this by considering separately every kind of element in $V$. For $u \in \{z_1, z_2\}$, we get
\begin{align*}
	f(\{u\})
	={} &
	1 - \left(1 - \frac{1 - \alpha}{n}\right)^{|\{u\} \cap X|}\left[\alpha\left(2 - \frac{|\{u\} \cap \{z_1, z_2\}|}{1 + 2\eps}\right)\right.\\
	&\mspace{100mu}+ \left.(1 - 2\alpha)\left(1 - \frac{|\{u\} \cap \{w\}|}{1 + 2\eps}\right)\left(1 - \frac{\alpha(1 - 2\alpha)^{-1} - 1}{n}\right)^{|\{u\} \cap Y|}\right]\\
	={}&
	1 - \alpha\left(2 - \frac{|\{u\}|}{1 + 2\eps}\right) - (1 - 2\alpha)
	=
	\frac{\alpha}{1 + 2\eps}
	\enspace.
\end{align*}

Consider now the case of $u = w$. In this case,
\begin{align*}
	f(\{w\})
	={}&
	1 - \left(1 - \frac{1 - \alpha}{n}\right)^{|\{w\} \cap X|}\left[\alpha\left(2 - \frac{|\{w\} \cap \{z_1, z_2\}|}{1 + 2\eps}\right)\right.\\
	&\mspace{100mu}+ \left.(1 - 2\alpha)\left(1 - \frac{|\{w\} \cap \{w\}|}{1 + 2\eps}\right)\left(1 - \frac{\alpha(1 - 2\alpha)^{-1} - 1}{n}\right)^{|\{w\} \cap Y|}\right]\\
	={}&
	1 - 2\alpha - (1 - 2\alpha)\left(1 - \frac{|\{w\}|}{1 + 2\eps}\right)
	=
	\frac{1 - 2\alpha}{1 + 2\eps}
	\leq
	\frac{\alpha}{1 + 2\eps}
	\enspace,
\end{align*}
where the last inequality holds since we assume $\alpha \geq 1/3$.

Next, consider the case of $u \in X$. In this case,
\begin{align*}
	f(\{u\})
	={}&
	1 - \left(1 - \frac{1 - \alpha}{n}\right)^{|\{u\} \cap X|}\left[\alpha\left(2 - \frac{|\{u\} \cap \{z_1, z_2\}|}{1 + 2\eps}\right)\right.\\
	&\mspace{100mu}+ \left.(1 - 2\alpha)\left(1 - \frac{|\{u\} \cap \{w\}|}{1 + 2\eps}\right)\left(1 - \frac{\alpha(1 - 2\alpha)^{-1} - 1}{n}\right)^{|\{u\} \cap Y|}\right]\\
	={}&
	1 - \left(1 - \frac{1 - \alpha}{n}\right)^{|\{u\}|}[2\alpha + (1 - 2\alpha)]
	=
	\frac{1 - \alpha}{n}
	\leq
	\frac{\alpha}{1 + 2\eps}
	\enspace,
\end{align*}
where the last inequality holds for $\eps \leq 1/2$ and $n \geq 4$ since $1 - \alpha \leq 2/3 \leq 2\alpha$.

Finally, consider the case of $u \in Y$. In this case,
\begin{align*}
	f(\{u\})
	={}&
	1 - \left(1 - \frac{1 - \alpha}{n}\right)^{|\{u\} \cap X|}\left[\alpha\left(2 - \frac{|\{u\} \cap \{z_1, z_2\}|}{1 + 2\eps}\right)\right.\\
	&\mspace{100mu}+ \left.(1 - 2\alpha)\left(1 - \frac{|\{u\} \cap \{w\}|}{1 + 2\eps}\right)\left(1 - \frac{\alpha(1 - 2\alpha)^{-1} - 1}{n}\right)^{|\{u\} \cap Y|}\right]\\
	={}&
	1 - 2\alpha - (1 - 2\alpha)\left(1 - \frac{\alpha(1 - 2\alpha)^{-1} - 1}{n}\right)^{|\{u\}|}
	=
	\frac{\alpha - (1 - 2\alpha)}{n}
	\leq
	\frac{\alpha}{1 + 2\eps}
	\enspace,
\end{align*}
where the last inequality holds for $\eps \leq 1/2$ and $n \geq 2$ since $1 - 2\alpha \geq 1 - 2 \cdot 0.49 = 0.02$.
\end{proof}

The next three claims are devoted to analyzing the value of the output of {\PG} given our example instance. In particular, they show that, as explained in the intuition given above, {\PG} outputs the set $X \cup Y$ as its solution. 
\begin{lemma} \label{lem:X_select}
Given that the set of elements selected so far by {\sc Plain Greedy} is $S \subset X$, the next element selected by {\sc Plain Greedy} is another element of $X$.
\end{lemma}
\begin{proof}
First of all, we observe that the total cost of all the elements of $X$ together is
\[
	c(X)
	=
	n \cdot \frac{(1 - \alpha)(1 + \eps)}{n}
	=
	(1 - \alpha)(1 + \eps)
	\leq
	1
	\enspace,
\]
where the inequality holds for a small enough $\eps$ because $\alpha \geq 1/3$. Therefore, to prove the lemma, it suffices to calculate the density of every element $u \in V \setminus S$ with respect to $S$, and show that the density of the elements of $X \setminus S$ is the largest. The density of an element $u \in X \setminus S$ is
\begin{align} \label{eq:X_leading_density}
	\frac{f(u \mid S)}{c(u)}
	={}&
	\frac{\left\{1 - \left(1 - \frac{1 - \alpha}{n}\right)^{|S| + 1}[2\alpha + (1 - 2\alpha)]\right\} - \left\{1 - \left(1 - \frac{1 - \alpha}{n}\right)^{|S|}[2\alpha + (1 - 2\alpha)]\right\}}{(1 - \alpha)(1 + \eps)/n}\\ \nonumber
	={} &
	\left(1 - \frac{1 - \alpha}{n}\right)^{|S|} \cdot \frac{\left\{1 - \left(1 - \frac{1 - \alpha}{n}\right)\right\}}{(1 - \alpha)(1 + \eps)/n}
	=
	\frac{1}{1 + \eps} \cdot \left(1 - \frac{1 - \alpha}{n}\right)^{|S|}
	\enspace.
\end{align}

Let us now calculate the density of an element $u \in Y$ with respect to $S$.
\begin{align*}
	\frac{f(u \mid S)}{c(u)}\mspace{-1mu}&\mspace{1mu}
	=
	\frac{\left\{1 - \left(1 - \frac{1 - \alpha}{n}\right)^{|S|}\left[2\alpha + (1 - 2\alpha)\left(1 - \frac{\alpha(1 - 2\alpha)^{-1} - 1}{n}\right)\right]\right\}}{(3\alpha - 1)(1 + 1.5\eps)/n} \\&\mspace{220mu}- \frac{\left\{1 - \left(1 - \frac{1 - \alpha}{n}\right)^{|S|}[2\alpha + (1 - 2\alpha)]\right\}}{(3\alpha - 1)(1 + 1.5\eps)/n}\\
	={} &
	\left(1 - \frac{1 - \alpha}{n}\right)^{|S|} \cdot \frac{(1 - 2\alpha) \cdot [\alpha(1 - 2\alpha)^{-1} - 1] / n}{(3\alpha - 1)(1 + 1.5\eps)/n}
	=
	\frac{1}{1 + 1.5\eps} \cdot \left(1 - \frac{1 - \alpha}{n}\right)^{|S|}
	\enspace,
\end{align*}
which is clearly smaller than \eqref{eq:X_leading_density}.

Next, the density of an element $u \in \{z_1, z_2\}$ with respected to $S$ is
\begin{align*}
	\frac{f(u \mid S)}{c(u)}
	={}&
	\frac{\left\{1 - \left(1 - \frac{1 - \alpha}{n}\right)^{|S|}\left[\alpha\left(2 - \frac{1}{1 + 2\eps}\right) + (1 - 2\alpha)\right]\right\}}{\alpha} \\&\mspace{200mu}- \frac{\left\{1 - \left(1 -  \frac{1 - \alpha}{n}\right)^{|S|}[2\alpha + (1 - 2\alpha)]\right\}}{\alpha}\\
	={} &
	\left(1 - \frac{1 - \alpha}{n}\right)^{|S|} \cdot \frac{\alpha / (1 + 2\eps)}{\alpha}
	=
	\frac{1}{1 + 2\eps} \cdot \left(1 - \frac{1 - \alpha}{n}\right)^{|S|}
	\enspace,
\end{align*}
which is again clearly smaller than \eqref{eq:X_leading_density}.

Finally, the density of the element $w$ with respected to $S$ is
\begin{align*}
	\frac{f(w \mid S)}{c(w)}
	={}&
	\frac{\left\{1 - \left(1 - \frac{1 - \alpha}{n}\right)^{|S|}\left[\alpha + (1 - 2\alpha)\left(1 - \frac{1}{1 + 2\eps}\right)\right]\right\}}{1 - 2\alpha} \\&\mspace{200mu}- \frac{\left\{1 - \left(1 -\frac{1 - \alpha}{n}\right)^{|S|}[2\alpha + (1 - 2\alpha)]\right\}}{1 - 2\alpha}\\
	={} &
	\left(1 - \frac{1 - \alpha}{n}\right)^{|S|} \cdot \frac{(1 - 2\alpha) / (1 + 2\eps)}{1 - 2\alpha}
	=
	\frac{1}{1 + 2\eps} \cdot \left(1 - \frac{1 - \alpha}{n}\right)^{|S|}
	\enspace,
\end{align*}
which is also smaller than \eqref{eq:X_leading_density}.
\end{proof}

\begin{lemma} \label{lem:Y_select}
Given that the set of elements selected so far by {\sc Plain Greedy} is $X \subseteq S \subset X \cup Y$, the next element selected by {\sc Plain Greedy} is another element of $Y$.
\end{lemma}
\begin{proof}
First of all, we observe that the total cost of all the elements of $X \cup Y$ is
\begin{align*}
	c(X \cup Y)
	={} &
	n \cdot \frac{(1 - \alpha)(1 + \eps)}{n} + n \cdot \frac{(3\alpha - 1)(1 + 1.5\eps)}{n}\\
	={} &
	(1 - \alpha)(1 + \eps) + (3\alpha - 1)(1 + 1.5\eps)
	=
	2\alpha + \eps(7\alpha - 1)/2
	\leq
	1
	\enspace,
\end{align*}
where the inequality holds for a small enough $\eps$ because $\alpha \leq 0.49$. Therefore, adding another element of $Y$ to $S$ does not violate feasibility. In contrast, for every $i \in \{1, 2\}$,
\[
	c(S + z_i)
	\geq
	c(X) + c(z_i)
	=
	n \cdot \frac{(1 - \alpha)(1 + \eps)}{n} + \alpha
	>
	(1 - \alpha) + \alpha
	=
	1
	\enspace,
\]
and thus, the elements $z_1$ and $z_2$ cannot be added to $S$ because of the knapsack constraint. Therefore, to prove the lemma, it suffices to calculate the density of every element $u \in V \setminus (S \cup \{z_1, z_2\})$ with respect to $S$, and show that the density of the elements of $Y \setminus S$ is the largest. The density of an element $u \in Y \setminus S$ is
\begin{align} \label{eq:Y_leading_density}
	\frac{f(u \mid S)}{c(u)}\mspace{-27mu}&\mspace{27mu}
	=
	\frac{\left\{1 - \left(1 - \frac{1 - \alpha}{n}\right)^{|X|}\left[2\alpha + (1 - 2\alpha)\left(1 - \frac{\alpha(1 - 2\alpha)^{-1} - 1}{n}\right)^{|S \cap Y| + 1}\right]\right\}}{(3\alpha - 1)(1 + 1.5\eps)/n} \\\nonumber&\mspace{100mu}- \frac{\left\{1 - \left(1 - \frac{1 - \alpha}{n}\right)^{|X|}\left[2\alpha + (1 - 2\alpha)\left(1 - \frac{\alpha(1 - 2\alpha)^{-1} - 1}{n}\right)^{|S \cap Y|}\right]\right\}}{(3\alpha - 1)(1 + 1.5\eps)/n}\\ \nonumber
	={} &
	\left(\mspace{-1mu}1 - \frac{1 - \alpha}{n}\mspace{-1mu}\right)^{\mspace{-2mu}|X|}\left(\mspace{-1mu}1 - \frac{\alpha(1 - 2\alpha)^{-1} - 1}{n}\mspace{-1mu}\right)^{\mspace{-2mu}|S \cap Y|} \cdot \frac{(1 - 2\alpha) \cdot \left\{\mspace{-1mu}1 - \left(\mspace{-1mu}1 - \frac{\alpha(1 - 2\alpha\mspace{-1mu})^{-1} - 1}{n}\right)\mspace{-1mu}\right\}}{(3\alpha - 1)(1 + 1.5\eps)/n}\\ \nonumber
	={} &
	\frac{1}{1 + 1.5\eps} \cdot \left(1 - \frac{1 - \alpha}{n}\right)^{|X|}\left(1 - \frac{\alpha(1 - 2\alpha)^{-1} - 1}{n}\right)^{|S \cap Y|}
	\enspace.
\end{align}

Consider now the element $w$, which is the only element of $V \setminus (S \cup \{z_1, z_2\})$ which does not belong to $Y$. The density of $w$ with respected to $S$ is
\begin{align*}
	\frac{f(w \mid S)}{c(w)}
	={} &
	\frac{\left\{1 - \left(1 - \frac{1 - \alpha}{n}\right)^{|X|}\left[\alpha + (1 - 2\alpha)\left(1 - \frac{1}{1 + 2\eps}\right)\left(1 - \frac{\alpha(1 - 2\alpha)^{-1} - 1}{n}\right)^{|S \cap Y|}\right]\right\}}{1 - 2\alpha} \\\nonumber&- \frac{\left\{1 - \left(1 - \frac{1 - \alpha}{n}\right)^{|X|}\left[2\alpha + (1 - 2\alpha)\left(1 - \frac{\alpha(1 - 2\alpha)^{-1} - 1}{n}\right)^{|S \cap Y|}\right]\right\}}{1 - 2\alpha}\\ \nonumber
	={} &
	\left(\mspace{-1mu}1 - \frac{1 - \alpha}{n}\mspace{-1mu}\right)^{\mspace{-2mu}|X|}\left(\mspace{-1mu}1 - \frac{\alpha(1 - 2\alpha)^{-1} - 1}{n}\mspace{-1mu}\right)^{\mspace{-2mu}|S \cap Y|} \cdot \frac{(1 - 2\alpha) \cdot \left\{\mspace{-1mu}1 - \left(\mspace{-1mu}1 - \frac{1}{1 + 2\eps}\mspace{-1mu}\right)\mspace{-1mu}\right\}}{1 - 2\alpha}\\ \nonumber
	={} &
	\frac{1}{1 + 2\eps} \cdot \left(1 - \frac{1 - \alpha}{n}\right)^{|X|}\left(1 - \frac{\alpha(1 - 2\alpha)^{-1} - 1}{n}\right)^{|S \cap Y|}
	\enspace,
\end{align*}
which is clearly smaller than \eqref{eq:Y_leading_density}.
\end{proof}

\begin{corollary} \label{cor:output_simple}
Given our example instance, the output set of {\PG} is $X \cup Y$, whose value is $f(X \cup Y) \leq 1 - e^{\alpha - 1} \left[2\alpha + (1 - 2\alpha)e^{1 - \alpha(1 - 2\alpha)^{-1}}\right] + \frac{577}{n}$.
\end{corollary}
\begin{proof}
Lemmata~\ref{lem:X_select} and~\ref{lem:Y_select} imply that the first $2n$ elements that {\sc Plain Greedy} selects are the elements of $X \cup Y$, whose total cost is
\begin{align*}
	c(X \cup Y)
	={} &
	n \cdot \frac{(1 - \alpha)(1 + \eps)}{n} + n \cdot \frac{(3\alpha - 1)(1 + 1.5\eps)}{n}\\
	={} &
	(1 - \alpha)(1 + \eps) + (3\alpha - 1)(1 + 1.5\eps)
	=
	2\alpha + \eps(7\alpha - 1)/2 > 2\alpha
	\enspace.
\end{align*}
Hence, the budget remaining at this point is less than $1 - 2\alpha \leq \alpha$, and thus, is smaller than the cost of any one of the remaining elements of $V$ (which are $z_1$, $z_2$ and $w$). Therefore, {\PG} cannot further increase the solution $X \cup Y$, and has to output it.

The value of the set $X \cup Y$ is
\begin{align*}
	f(X \cup Y)
	={}&
	1 - \left(1 - \frac{1 - \alpha}{n}\right)^{|(X \cup Y) \cap X|}\left[\alpha\left(2 - \frac{|(X \cup Y) \cap \{z_1, z_2\}|}{1 + 2\eps}\right)\right.\\
	&+ \left.(1 - 2\alpha)\left(1 - \frac{|(X \cup Y) \cap \{w\}|}{1 + 2\eps}\right)\left(1 - \frac{\alpha(1 - 2\alpha)^{-1} - 1}{n}\right)^{|(X \cup Y) \cap Y|}\right]\\
	={}&
	1 - \left(1 - \frac{1 - \alpha}{n}\right)^{|X|} \left[2\alpha + (1 - 2\alpha)\left(1 - \frac{\alpha(1 - 2\alpha)^{-1} - 1}{n}\right)^{|Y|}\right]\\
	\leq{} &
	1 - e^{\alpha - 1}\left(1 - \frac{1}{n}\right) \left[2\alpha + (1 - 2\alpha)e^{1 - \alpha(1 - 2\alpha)^{-1}}\left(1 - \frac{576}{n}\right)\right]\\
	\leq{} &
	1 - e^{\alpha - 1} \left[2\alpha + (1 - 2\alpha)e^{1 - \alpha(1 - 2\alpha)^{-1}}\right] + \frac{577}{n}
	\enspace,
\end{align*}
where the first inequality holds for $n \geq 24$ since $\alpha(1 - 2\alpha)^{-1} - 1 \leq 24$.
\end{proof}

We are now ready to prove the upper bound on the approximation ratio of {\Gr} given by \autoref{t:3}.
\begin{proposition}
The approximation ration of {\Gr} is no better than $0.462$.
\end{proposition}
\begin{proof}
Consider our example instance with the parameter $\alpha$ set to $0.461$. By Corollary~\ref{cor:output_simple}, given this instance {\PG} outputs the set $X \cup Y$, whose value is
\[
	f(X \cup Y)
	\leq
	1 - e^{\alpha - 1} \left[2\alpha + (1 - 2\alpha)e^{1 - \alpha(1 - 2\alpha)^{-1}}\right] + \frac{577}{n}
	\leq
	0.4618
	\leq
	\frac{0.462}{1 + 2\eps}
	\enspace,
\]
where the second inequality holds for a large enough $n$, and the last inequality holds for a small enough $\eps$. Since the value of the largest feasible singleton set is at most $\alpha / (1 + 2\eps) \leq 0.462 / (1 + 2\eps)$ by Lemma~\ref{lem:singleton_value}, and {\sc Greedy} outputs either a singleton set or the output $X \cup Y$ of {\sc Plain Greedy}, the value of the output of {\sc Greedy} for our example instance is at most $0.462 / (1 + 2\eps)$ as well.

In contrast, by \autoref{obs:O_properties}, the set $O$ is a feasible solution of value $1 / (1 + 2\eps)$, and therefore, the approximation ratio of {\Gr} when it is given our example instance as input is no better than
$\frac{0.462 / (1 + 2\eps)}{1 / (1 + 2\eps)}=0.462$.
\end{proof}}

\appendix
\inFull{
\inFull{\section{Missing Proofs of Section~\ref{s:1_guess}}}
\inConference{\subsection{Missing Proofs of Section~\ref{s:1_guess}}}
\label{app:G_plus_missing}

In this section we give the proofs that have been omitted from Section~\ref{s:1_guess}.

\inConference{
\tHalf*
\proofTHalf
}

\lz*
\begin{proof}
To see that the first part of the lemma holds for $y \in (0, 1/2]$, note that as $z(y)$ increases from $y$ to $1/2$, 
the left the hand side of the inequality defining $z(y)$ continuously decreases from $0$ to $2y - 1$, 
and the right hand side of this inequality continuously increases to $\ln\left(\frac{1/2}{1 - y}\right)$, 
which is less than $0$ and at least $- \ln(2 - 2y) \geq -[e^{\ln(2 - 2y)} - 1] = -(1 - 2y) = 2y - 1$. 
To see that the first part of the lemma holds for $y = 0$ as well, note that by definition $z(0)$ satisfies 
$-1 = \ln(z(0))$, and $1/e \in [0, 1/2]$ is the only number satisfying this equality.

Up to this point we have proved that the function $z(y)$ is well defined. Our next objective is to show that it is also deferentiable. To do that, consider the function $G(y, z) = y / z - 1 - \ln\left(\frac{z}{1 - y}\right)$ and a point $(y, z)$ obeying $y \in [0, 1/2]$ and $z = z(y)$. Clearly, $G(y, z) = 0$, and the derivative $dG(y, z)/dz$ obeys at this point
\[
	\frac{dG(y, z)}{dz}
	=
	-\frac{y}{z^2} - \frac{1 / (1 - y)}{z/(1 - y)}
	=
	-\frac{y}{z^2} - \frac{1}{z}
	<
	0
	\enspace.
\]
Hence, by the implicit function theorem, $z(y)$ is a continuous differentiable function of $y$ for this range of $y$.

Given the knowledge that the derivative of $z(y)$ exists within the range $y \in [0, 1/2]$, we can now calculate it by 
taking the derivative with respect to $y$ of both sides of the inequality defining $z(y)$. Doing so yields
\[
\frac{z(y) - y \cdot z'(y)}{z^2(y)} = \frac{[z'(y) \cdot (1 - y) + z(y)] / (1 - y)^2}{z(y) / (1 - y)} \enspace,
\]
and solving for $z'(y)$ gives
\[
z'(y) = \frac{z(y) (1-y-z(y))}{(1-y)(z(y) + y)} \geq 0 \enspace,
\]
where the inequality holds since the first part of the lemma shows that $z(y) \in (0, 1/2] \subseteq (0,1-y]$. Note that, since the derivative $z'(y)$ is non-negative, we get that $z(y)$ is a non-decreasing function, as promised.
\end{proof}

\inConference{
\lemLargeW*
\proofLemaLargeW

\lPh*
\proofLPh
}

\lemP*
\begin{proof}
To make our calculations easier, it is useful to define $y=\f{x}{1-x}$, which implies $x=\f{y}{1+y}$ and $1-x=\f{1}{1+y}$, and therefore, also
\[
	p(x)=x+(1 -x) \cdot \left(1-z\left(\frac{x}{1 - x}\right)\right)=\f{y}{1+y}+\f{1-z(y)}{1+y}=1-\f{z(y)}{1+y}
	\enspace.
\]
We denote the rightmost side of the last equality by $q(y)$.
Since $y$ goes exactly over all the values of the range $[0, 1/2]$ when $x$ grows from $0$ to $1/3$, $\min\{p(x):x \in [0,1/3]\}=\min\{q(y):y \in [0,1/2]\}$.
Thus, to prove the lemma it suffices to show that $\min\{q(y):y \in [0,1/2]\} = \f{3-\ln 4}{4-\ln 4}$, which we do in the rest of this proof.

Let us now use the shorthand $z = z(y)$. Using this shorthand, we get
\begin{align*}
	q'(y)
	={} &
	-\f{z'(y)(1+y)-z}{(1+y)^2}
	=
	\f{z - \frac{z(1 + y) (1 - y - z)}{(1 - y)(z + y)}}{(1+y)^2}
	=
	\frac{z}{(1 + y)^2} \cdot \left(1 - \frac{(1 + y) (1 - y - z)}{(1 - y)(z + y)}\right)\\
	={} &
	\frac{z}{(1 + y)^2} \cdot \left(\frac{2(z + y) - (1 + y)}{(1 - y)(z + y)}\right)
	=
	\frac{z}{(1 + y)^2(1 - y)(z + y)} \cdot (2z + y - 1)
	\enspace,
\end{align*}
where the second equality follow from Lemma~\ref{l:z}. Since Lemma~\ref{l:z} guarantees that $z \geq z(0) = 1/e$, the sign of $q'(y)$ (within the range $[0, 1/2]$) is equal to the sign of $2z + y - 1$, and the last expression is a (strictly) increasing function of $y$ since $z = z(y)$ is a non-decreasing function of $y$ by Lemma~\ref{l:z}. Hence, $q(y)$ is a convex function within this range which takes its minimum value at the point $y_0$ in which $q'(y_0) = 0
$ (assuming there is such a point). In other words, to complete the proof of the lemma it remains to show that there is a point $y_0 \in [0, 1/2]$ obeying $q'(y_0) = 0$ and $q(y_0) = \f{3-\ln 4}{4-\ln 4}$.

We show that $y_0 = \f{1-\ln 2}{3-\ln 2} \in [0, 1/3]$ has the above mentioned properties. According to Lemma~\ref{l:z}, $z(y_0)$ is the sole value in $[y_0, 1/2]$ obeying the inequality
\[
	\f{y_0}{z(y_0)} - 1 = \ln\left(\frac{z(y_0)}{1 - y_0}\right)
	\enspace,
\]
and one can verify that $z(y_0) = (1 - y_0)/2 \in [y_0, 1/2]$ obeys this inequality. Hence,
\begin{align*}
	q(y_0)
	={} &
	1-\f{z(y_0)}{1+y_0}
	=
	1-\f{(1 - y_0)/2}{1+y_0}
	=
	1 - \frac{1 - (1 - \ln 2) / (3 - \ln 2)}{2(1 + (1 - \ln 2) / (3 - \ln 2))}\\
	={} &
	1 - \frac{(3 - \ln 2) - (1 - \ln 2)}{2((3 - \ln 2) + (1 - \ln 2))}
	=
	1 - \frac{1}{4 - 2\ln 2}
	=
	\frac{3 - \ln 4}{4 - \ln 4}
\end{align*}
and
\begin{align*}
	q'(y_0)
	={} &
	\frac{z(y_0)}{(1 + y_0)^2(1 - y_0)(z(y_0) + y_0)} \cdot (2z(y_0) + y_0 - 1)\\
	={} &
	\frac{z(y_0)}{(1 + y_0)^2(1 - y_0)(z(y_0) + y_0)} \cdot ((1 - y_0) + y_0 - 1)
	=
	0
	\enspace.
	\qedhere
\end{align*}
\end{proof}

}

\bibliography{cov-bib}

\begin{thebibliography}{10}

\bibitem{BV}
A.~Badanidiyuru and J.~Vondr\'{a}k.
\newblock Fast algorithms for maximizing submodular functions.
\newblock In {\em SODA}, pages 1497--1514, 2014.

\bibitem{calinescu2011maximizing}
Gruia C{\u{a}}linescu, Chandra Chekuri, Martin P{\'{a}}l, and Jan
  Vondr{\'{a}}k.
\newblock Maximizing a monotone submodular function subject to a matroid
  constraint.
\newblock {\em {SIAM} J. Comput.}, 40(6):1740--1766, 2011.
\newblock \href {https://doi.org/10.1137/080733991}
  {\path{doi:10.1137/080733991}}.

\bibitem{CK}
R.~Cohen and L.~Katzir.
\newblock The generalized maximum coverage problem.
\newblock {\em Inf. Process. Lett.}, 108(1):15--22, 2008.

\bibitem{de2020regression}
Abir De, Paramita Koley, Niloy Ganguly, and Manuel Gomez{-}Rodriguez.
\newblock Regression under human assistance.
\newblock In {\em AAAI}, pages 2611--2620, 2020.
\newblock URL: \url{https://aaai.org/ojs/index.php/AAAI/article/view/5645}.

\bibitem{elenberg2017streaming}
Ethan~R. Elenberg, Alexandros~G. Dimakis, Moran Feldman, and Amin Karbasi.
\newblock Streaming weak submodularity: Interpreting neural networks on the
  fly.
\newblock In {\em NeurIPS}, pages 4047--4057, 2017.

\bibitem{EN}
A.~Ene and H.~L. Nguyen.
\newblock A nearly-linear time algorithm for submodular maximization with a
  knapsack constraint.
\newblock In {\em ICALP}, pages 53:1--53:12, 2019.

\bibitem{feldman2011improved}
Moran Feldman, Joseph Naor, Roy Schwartz, and Justin Ward.
\newblock Improved approximations for k-exchange systems - (extended abstract).
\newblock In Camil Demetrescu and Magn{\'{u}}s~M. Halld{\'{o}}rsson, editors,
  {\em ESA}, volume 6942 of {\em Lecture Notes in Computer Science}, pages
  784--798. Springer, 2011.
\newblock \href {https://doi.org/10.1007/978-3-642-23719-5\_66}
  {\path{doi:10.1007/978-3-642-23719-5\_66}}.

\bibitem{filmus2014monotone}
Yuval Filmus and Justin Ward.
\newblock Monotone submodular maximization over a matroid via non-oblivious
  local search.
\newblock {\em {SIAM} J. Comput.}, 43(2):514--542, 2014.
\newblock \href {https://doi.org/10.1137/130920277}
  {\path{doi:10.1137/130920277}}.

\bibitem{kazemi2018scalable}
Ehsan Kazemi, Morteza Zadimoghaddam, and Amin Karbasi.
\newblock Scalable deletion-robust submodular maximization: Data summarization
  with privacy and fairness constraints.
\newblock In {\em ICML}, pages 2549--2558, 2018.

\bibitem{KMN}
S.~Khuller, A.~Moss, and J.~Naor.
\newblock The budgeted maximum coverage problem.
\newblock {\em Inform. Process. Lett.}, 70:39--45, 1999.

\bibitem{kulik2013approximations}
Ariel Kulik, Hadas Shachnai, and Tami Tamir.
\newblock Approximations for monotone and nonmonotone submodular maximization
  with knapsack constraints.
\newblock {\em Math. Oper. Res.}, 38(4):729--739, 2013.
\newblock \href {https://doi.org/10.1287/moor.2013.0592}
  {\path{doi:10.1287/moor.2013.0592}}.

\bibitem{lee2010submodular}
Jon Lee, Maxim Sviridenko, and Jan Vondr{\'{a}}k.
\newblock Submodular maximization over multiple matroids via generalized
  exchange properties.
\newblock {\em Math. Oper. Res.}, 35(4):795--806, 2010.
\newblock \href {https://doi.org/10.1287/moor.1100.0463}
  {\path{doi:10.1287/moor.1100.0463}}.

\bibitem{lei2019discrete}
Qi~Lei, Lingfei Wu, Pin{-}Yu Chen, Alex Dimakis, Inderjit~S. Dhillon, and
  Michael~J. Witbrock.
\newblock Discrete adversarial attacks and submodular optimization with
  applications to text classification.
\newblock In {\em MLSys}, pages 146--165, 2019.
\newblock URL: \url{https://proceedings.mlsys.org/book/284.pdf}.

\bibitem{libbrecht2018choosing}
Maxwell~W. Libbrecht, Jeffrey~A. Bilmes, and William~Stafford Noble.
\newblock Choosing non-redundant representative subsets of protein sequence
  data sets using submodular optimization.
\newblock {\em Proteins: Structure, Function, and Bioinformatics},
  86(4):454--466, 2018.

\bibitem{mirzasoleiman2016distributed}
Baharan Mirzasoleiman, Amin Karbasi, Rik Sarkar, and Andreas Krause.
\newblock Distributed submodular maximization.
\newblock {\em Journal of Machine Learning Research}, 17:238:1--238:44, 2016.

\bibitem{mitrovic2018data}
Marko Mitrovic, Ehsan~Kazemi andMorteza Zadimoghaddam, and Amin Karbasi.
\newblock Data summarization at scale: {A} two-stage submodular approach.
\newblock In {\em ICML}, pages 3593--3602, 2018.

\bibitem{NWF}
G.~L. Nemhauser, L.~A. Wolsey, and M.~L. Fisher.
\newblock An analysis of approximations for maximizing submodular set
  functions-i.
\newblock {\em Math. Programming}, 14:265--294, 1978.

\bibitem{nemhauser1978best}
George~L. Nemhauser and Laurence~A. Wolsey.
\newblock Best algorithms for approximating the maximum of a submodular set
  function.
\newblock {\em Math. Oper. Res.}, 3(3):177--188, 1978.
\newblock \href {https://doi.org/10.1287/moor.3.3.177}
  {\path{doi:10.1287/moor.3.3.177}}.

\bibitem{salehi2017submodular}
Mehraveh Salehi, Amin Karbasi, Dustin Scheinost, and R.~Todd Constable.
\newblock A submodular approach to create individualized parcellations of the
  human brain.
\newblock In {\em MICCAI}, pages 478--485, 2017.

\bibitem{Svir}
M.~Sviridenko.
\newblock A note on maximizing a submodular set function subject to a knapsack
  constraint.
\newblock {\em Operations Research Letters}, 32:41--43, 2004.

\bibitem{TTLH}
J.~Tang, X.~Tang, A.~Lim, K.~Han, C.~Li, and J.~Yuan.
\newblock Revisiting modified greedy algorithm for monotone submodular
  maximization with a knapsack constraint.
\newblock {\em CoRR}, abs/2008.05391, 2020.
\newblock URL: \url{https://arxiv.org/abs/2008.05391}, \href
  {http://arxiv.org/abs/2008.05391} {\path{arXiv:2008.05391}}.

\bibitem{ward2012approximation}
Justin Ward.
\newblock A (k+3)/2-approximation algorithm for monotone submodular k-set
  packing and general k-exchange systems.
\newblock In Christoph D{\"{u}}rr and Thomas Wilke, editors, {\em STACS},
  volume~14 of {\em LIPIcs}, pages 42--53. Schloss Dagstuhl - Leibniz-Zentrum
  f{\"{u}}r Informatik, 2012.
\newblock \href {https://doi.org/10.4230/LIPIcs.STACS.2012.42}
  {\path{doi:10.4230/LIPIcs.STACS.2012.42}}.

\bibitem{W-B}
L.~A. Wolsey.
\newblock Maximising real-valued submodular functions: primal and dual
  heuristics for location problems.
\newblock {\em Math. Oper. Res.}, 7:410–425, 1982.

\bibitem{YZA}
G.~Yaroslavtsev, S.~Zhou, and D.~Avdiukhin.
\newblock ``bring your own greedy''+max: Near-optimal 1/2-approximations for
  submodular knapsack.
\newblock In {\em AISTATS}, pages 3263--3274, 2020.

\end{thebibliography}

\inConference{
\section{Missing Proofs}

\subsection{Missing Proof of Section~\ref{s:notation}} \label{app:missing_preliminaries}

In this section we give the proof that has been omitted from Section~\ref{s:notation}.

\lr*
\lrProof

\subsection{Missing Proof of Section~\ref{s:2_guess}} \label{app:two_guesses_missing}

In this section we give the proof that has been omitted from Section~\ref{s:2_guess}.

\lemOneElementLessGuarantee*
\proofLemOneElementLessGuarantee

}

\end{document}